\documentclass[a4paper,UKenglish,cleveref, autoref, thm-restate, nolinenumbers]{lites-v2021}
\nolinenumbers

\pdfoutput=1 
\hideLITES  

\usepackage{cite}
\usepackage{amsmath,amssymb,amsfonts}
\usepackage{amsthm}
\usepackage{graphicx}
\usepackage{textcomp}
\usepackage{xcolor}
\usepackage{comment}

\usepackage{algorithm}
\usepackage{algpseudocode}
\usepackage{caption}
\usepackage{subcaption}
\algrenewcommand\algorithmicindent{1.2em}

\usepackage{tikz}
\newcommand*\circled[1]{\tikz[baseline=(char.base)]{
            \node[shape=circle,draw,inner sep=0.2pt] (char) {#1};}}

\title{Limited Preemption of the 3-Phase Task Model using Preemption Thresholds}

\author{Thilanka Thilakasiri}{KTH Royal Institute of Technology, Sweden}{{thilanka@kth.se}}{https://orcid.org/0000-0001-9363-3525}{}

\author{Matthias Becker}{KTH Royal Institute of Technology, Sweden}{{mabecker@kth.se}}{https://orcid.org/0000-0002-1276-3609}{}

\authorrunning{T. Thilakasiri and M. Becker} 

\Copyright{Thilanka Thilakasiri and Matthias Becker} 

\ccsdesc[500]{Computer systems organization~Real-time system architecture}
\ccsdesc[500]{Computer systems organization~Real-time systems}
\ccsdesc[500]{Computer systems organization~Embedded software}

\keywords{3-phase model, \and limited preemption, \and preemption thresholds} 

\Volume{1}
\Issue{1}
\Article{42}
\DateSubmission{Date of submission}
\DateAcceptance{Date of acceptance}
\DatePublished{Date of publishing}
\SectionAreaEditor{LITES section area editor}

\begin{document}

\maketitle


\begin{abstract}
Phased execution models are a well-known solution to tackle the unpredictability of today's complex COTS multi-core platforms. 
The semantics of these models dedicate phases for a task's execution and shared memory accesses. 
Memory phases are solely dedicated to load all necessary instructions and data to private local memory, and to write back the results of the computation. During execution phases, only the private local memory is accessed.
While non-preemptive execution phases utilize the local memory well, schedulability is reduced due to blocking. 
On the other hand, fully preemptive execution phases allow for better schedulability, but require local memory to be large enough to hold all tasks involved in preemption simultaneously. 
Limited preemption is a promising approach that provides moderation between non-preemptive and fully preemptive scheduling. 

In this paper, we propose using preemption thresholds to limit the number of preemptions to minimize local memory usage while maintaining schedulability.
We propose worst-case response-time and worst-case memory-requirement analyses for sporadic 3-phase tasks under partitioned fixed-priority scheduling with preemption thresholds.
We consider two distinct arbitration schemes for memory-phase scheduling between cores: task-priority-based and core-priority-based. 
We further show how the state-of-the-art algorithm to assign preemption thresholds can be applied to the considered task model.
Evaluations demonstrate that preemption thresholds can significantly reduce the memory usage compared to fully preemptive scheduling, while maintaining high schedulability ratios compared to non-preemptive scheduling.

\end{abstract}

\section{Introduction}
Modern multi-core processors offer high computational power and low energy consumption compared to their single-core counterparts. 
However, the use of Commercial off-the-shelf (COTS) multi-core platforms in time-critical applications remains challenging due to the unpredictability that emerges from their design. 
Several resources are shared between the different compute cores, such as the interconnect, caches, memories, or I/O.
As a result, the timing behavior of an application executing on one core can be interfered with by other unrelated applications that execute on different cores at the same time, which poses significant challenges to determining the application's worst-case behavior ~\cite {Wilhelm:SIES2012}.
The phased execution models address this problem by dividing task execution into dedicated phases for computation and shared memory access \cite{pellizzoni2011predictable, durrieu2014predictable}.
The 3-phase model has dedicated phases for a \emph{read}, \emph{execute}, and \emph{write} operations. 
The read phase is used to load all data and code necessary for execution into core-local memory. 
During the execution phase, only the core-local memory is accessed, followed by a write phase that is used to write back the results of the computation to shared memory. 
Thus, shared memory access is avoided during computation.
Additionally, contention when accessing the main memory is avoided by not allowing two memory phases to be scheduled simultaneously.

Typically, non-preemptive execution is used for the 3-phase model to easily achieve predictability and to efficiently utilize the local memory.
However, this can lead to low schedulability as high-priority tasks can be blocked by low-priority tasks. 
While allowing preemption can improve schedulability, realizing preemption for the phased execution model becomes non-trivial given its semantics and the limited size of local memories, i.e., the data stored in the local memory is solely used during the computation phase. Thus, handling the intermediate results of a preempted task that are only stored in the local memory becomes crucial.
Different approaches to implementing preemption for the phased execution model are explored~\cite{AER_preempt}, considering these challenges.
In this paper, we focus on the \emph{keep-in-core} method that stores the data of both the preempting task and the preempted task in the local memory during preemption, avoiding additional memory phases to write the local data of the preempted task to the shared memory and read them back when its execution is resumed. 
The keep-in-core approach does not allow task migration after preemption.

In this work, we focus on partitioned fixed-priority scheduling, where each task is statically assigned to a core. 
Thus, the keep-in-core approach is a natural fit.
With this preemption approach, both the preempting and preempted tasks' data must fit in the local memory.
Thus, to decide the schedulability of a task set, it is important to consider if each task is completed before its deadline, as well as if the local memory is sufficient to hold data of all tasks that can be involved in a preemption chain at any given time. 
Given that the typical local memory size in COTS platforms is limited, i.e., around 64 kB  \cite{NXPS32Z27, StellarE1}, this capacity can be easily exceeded with fully preemptive execution phases.
Limiting preemption to only where beneficial can therefore reduce the local memory usage while maintaining schedulability.




The study in~\cite{AER_preempt} explores different preemption strategies for 3-phase tasks under table-driven scheduling and shows that limited preemptive approaches perform best.
Yao et al. \cite{yao2012memory} and Melani et al. \cite{melani2017} also have highlighted possible benefits of limited preemption methods to address the drawbacks due to the limitation of the local memory. While a detailed investigation has been mentioned as future work, to the best of our knowledge, it has not been explored yet. This is the first paper to utilize limited preemption to improve the memory feasibility of 3-phase tasks.
Preemption thresholds (PT) can be used to reduce preemptions that are unnecessary to maintain schedulability, which is shown to reduce the stack space~\cite{Yao:EMSOFT2010, saksena2000scalable}.
Preemption thresholds can also be easily supported by operating systems~\cite{Becker:SIES2015}.
Thus, this paper addresses the online scheduling of 3-phase tasks under partitioned fixed-priority scheduling with preemption thresholds, aiming to reduce memory usage while maintaining schedulability. 

In detail, the paper has the following contributions:
\begin{itemize}
    \item  Schedulability tests for 3-phase tasks under partitioned fixed-priority scheduling with preemption thresholds, with a task-priority based memory-centric scheduler and a core-priority based memory-centric scheduler.
    \item An analysis to bound the maximum required local memory considering preemption thresholds.
    \item Our evaluations show the benefits of preemption thresholds for 3-phase tasks compared to both fully- and non-preemptive scheduling. 

\end{itemize}

\section{Related Work}
\label{sec:RelatedWork}

Timing verification of multi-core real-time systems is challenging due to the unpredictability of the underlying hardware platform. 
Consequently, this area has received significant attention in recent years~\cite{maiza2019survey, lugo2022survey}.
Among those, phased execution models are an established approach to containing the problem's complexity by resolving resource contention through scheduling in software. 
The PRedictable Execution Model (PREM) is initially proposed by Pellizzoni et al.~\cite{pellizzoni2008coscheduling} and later extended into the 3-phase~\cite{alhammad2014time} or Acquisition, Execution, Restitution (AER) model~\cite{durrieu2014predictable}.
For predictability, the phased execution model is typically explored under non-preemptive scheduling~\cite{maia2016closer, 7557865, exact_AER}.
A number of works also explore the phased execution model under different settings of preemptive scheduling.
The main focus of these works is on scheduling. Thus, the aspect of handling local memory is not considered.
Either the details on the management of local memory are not provided~\cite{huang2016mirror, senoussaoui2022contention} or it is assumed that the core's local memory is sufficiently large to store data of all the tasks assigned to it~\cite{yao2015global, yao2012memory, arora2022analyzing, sanudo2016schedulability, melani2017}.
There exist frameworks to divide the tasks into multiple segments if tasks do not fit in the local memory \cite{segmentTasks_soliman}. However, these task-splitting techniques must be done at design time, and tasks can only be split at specific places, considering the properties of their programs. 
Tabish et al. \cite{tabish2023x} use the segment streaming model proposed by Soliman et al. \cite{9052209}. Preemption of streaming segments is allowed in a partitioned system. In a case of preemption, the data of the preempted segment is written to the shared memory and read again when resuming its execution.
Thilakasiri et al.~\cite{AER_preempt} showed that due to the semantics of the phased execution models and the typically limited size of local memories in COTS platforms, handling local data during preemption is crucial.
They explored different alternative methods to realize preemption in the phased execution model.
Two methods are discussed. 
The waiting time minimizing preemption method (WMPM) or \emph{write-back} method schedules an additional write and read phase to store the intermediate data of the preempted task in global memory and to read the data back when its execution is resumed. The overhead minimizing preemption method (OMPM) or \emph{keep-in-core} method keeps the local data of the preempted task in local memory during preemption. 
Their results, based on table-driven scheduling, demonstrate that a limited preemptive approach performs best for either preemption method.
Similar indications are made by others~\cite{yao2012memory, melani2017, Wip_RTSS}, while a detailed investigation is left for future work.

Many works have proposed different limited preemption methods for conventional tasks \cite{Buttazzo2024, lee2012controlling} to achieve a middle path to reduce large memory requirements and preemption overheads in fully preemptive scheduling and to reduce blocking in non-preemptive scheduling.
Preemption points, deferred preemption, and preemption thresholds are well-known limited preemption techniques~\cite{Buttazzo2024}.
There are several works proposing methods for conventional tasks that combine different limited preemption methods, i.e., preemption points and preemption thresholds \cite{bril_PPandPT_1, bril_PPandPT_2} on single-core systems.
While a 3-phase task can be seen as equivalent to a special case of a conventional task with preemption points and preemption thresholds, the above-mentioned works can not be applied to 3-phase tasks on a multi-core system, as bounding the inter-core interference and blocking, which is both crucial and non-trivial, is missing.
For the 3-phase model, Arora et al. \cite{arora2022analyzing} proposed a response time analysis for fixed priority scheduling where preemption points are at the boundaries of the task phases. 
Application-specific preemption points can be selected by analysing the programs of tasks to improve the benefit of preemption points.
Though the objective is different, this has been explored by the works that exist in the literature for splitting 3-phase tasks \cite{segmentTasks_soliman}. 
Preemption thresholds limit the number of preemptions by limiting possible preempting tasks. 
The priority of a task is temporarily increased to a threshold, and preemption is only allowed from tasks that have a priority higher than this threshold~\cite{wang1999scheduling, davis2000embedded, saksena2000scalable, regehr2002scheduling, Yao:EMSOFT2010}. 
Different approaches address the configuration of preemption thresholds to reduce the stack requirements of applications~\cite{wang2016minimizing, saksena2000scalable}.
Though many works have addressed preemption thresholds for conventional tasks, the area has not been explored for phased execution models.

In this work, we address partitioned fixed-priority scheduling with preemption thresholds for the 3-phase task model under the keep-in-core preemption method, as it is best suited among the different limited preemption models to control memory requirements by selecting preemption thresholds.
We show how both worst-case response times and maximum local memory requirements can be computed under preemption thresholds and how threshold values can be configured to reach a schedulable system configuration.

\section{System Model}

We consider a multi-core system with $m$ cores, where $\Pi$ denotes the set of all cores and $\pi_i$ denotes the $i^{th}$ core in $\Pi$. 
Each core has its private, local memory, i.e., scratch-pad memory (SPM) or cache of size $S$. This local memory stores the data and instructions needed to execute the tasks assigned to its core. All cores can access the main memory via a shared interconnect, e.g., a bus.

\subsection{Application Model}

The considered application consists of a set of $n$ sporadic tasks $\tau = \{ \tau_1, \tau_2, ..., \tau_n\}$ that can be scheduled independently\footnote{This does not necessarily mean that the tasks are logically independent. Tasks can communicate via shared variables/memory following last-is-best semantics, which is, for example, typical in automotive systems \cite{kramer2015real, 7557865}.}.
Each task is divided into three phases, i.e., the read (R) phase, the execution (E) phase, and the write (W) phase, and are respectively dedicated to read, execute, and write operations of a task.
In the R-phase, all input data and instructions of the task are copied from the main memory to the local memory of the core. 
Thus, the following E-phase is executed using only local data copies of the task. 
Once the E-phase is completed, the W-phase writes the output data to the main memory. 
Each task $\tau_i$ is represented using the tuple $(T_i, D_i, M_i, P_i, \theta_i, C_i^r, C_i^e, C_i^w)$. 
$T_i$ and $D_i$ represent the task's minimum inter-arrival time between two consecutive releases and the relative deadline, respectively, where $D_i \leq T_i$. 
$M_i$ is the memory footprint of the task. It is considered to be the sum of the code size ($\mathit{Im_i}$), total data size ($\mathit{Dm_i}$), and the maximum stack usage ($\mathit{Ms_i}$) of the task, i.e., $\mathit{M_i = Im_i + Dm_i + Ms_i}$.
$P_i$ represents the task's fixed nominal priority and $\theta_i$ represents tasks's preemption threshold, where $P_i \leq \theta_i$. Higher values indicate higher priority.
$C_i^r$, $C_i^e$, and $C_i^w$ respectively represent the worst-case execution times for the read, execute, and write phases. 
The total worst-case execution time $C_i$ is their sum, i.e., $C_i = C_i^{r} + C_i^{e} + C_i^{w}$.
The utilization of a task can be represented as $u_i = C_i/ T_i$, where the total utilization of the task set, $U$ is $\sum_{i=1}^{n} u_i$.
Each release of a task is called a \emph{job} or an \emph{instance}. The $k^{th}$ job of task $\tau_i$ is represented by $\tau_{i,k}$.

The core task $\tau_i$ is assigned to is called the local core $\pi_l$, and other cores are referred to as remote cores, denoted by $\pi_r$ where $\mathit{\pi_l \neq \pi_r}$. 
Other tasks, $\tau_j$, i.e., $\mathit{\tau_j \neq \tau_i}$, that have a higher nominal priority, higher or equal nominal priority, and lower nominal priority are respectively represented by $\mathit{hp_{i}}$, $\mathit{hep_{i}}$, and $\mathit{lp_{i}}$. 



\subsection{Execution Model}


Partitioned fixed-priority scheduling is used. 
Thus, each task is statically assigned to a core, and migration is not allowed. 
Scheduling the assigned tasks locally on their respective core is done based on fixed task priorities.
Execution phases of tasks are preemptable.
We consider preemption thresholds (PT). The priority of a task $\tau_i$ is raised to its preemption threshold $\theta_i$ with the start of the read phase.
The execute phase of a task starts directly upon its read phase's completion.
Only tasks that have a nominal priority higher than the executing task's preemption threshold $\theta_i$ are allowed to preempt $\tau_i$'s execution phase. 
After the write phase finishes, the priority is restored to its nominal value.
The preemption threshold \emph{only} affects the local tasks that are mapped to the same core.
Hence, the PT of a task does not affect the scheduling of memory phases globally.

In this paper, our focus is on using preemption thresholds at the local, core-level. Different arbitration algorithms can be considered to schedule the memory phases globally on the shared interconnect, e.g., fixed task priority scheduling, fixed core-priority scheduling, and time division multiple access (TDMA). In this paper, we examine two cases, i.e., fixed task-priority and fixed core-priority memory-centric scheduling, and present worst-case response time (WCRT) analyses for them.
In task-priority memory-centric scheduling, the nominal task priority is used to schedule memory phases. Read and write memory phases are considered non-preemptive, similar to \cite{arora2022analyzing, memory_bound_RTNS}. 
In core-priority memory-centric scheduling, all the tasks assigned to the same core, $\pi_l$, are given the same memory access priority, $P^{mem}_l$, for scheduling the memory phases. The memory phases are considered to be preemptable by tasks assigned to other cores that have a higher memory access priority, but not by tasks assigned to the same core (i.e., that have the same memory priority), similar to \cite{schwaricke2020fixed}. 
In both task-priority and core-priority memory-centric schedulers, to avoid contention when accessing the shared memory, only one memory phase is allowed to be active at a time.


\begin{figure}[b]
    \centering
    \includegraphics[width=0.65\linewidth]{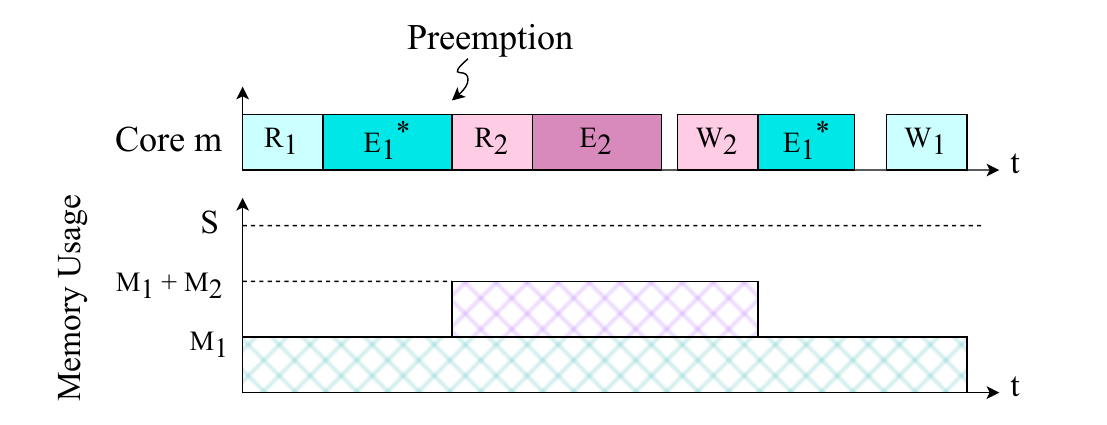}
    \caption{Execution and memory usage during preemption of $\tau_1$ (blue) by $\tau_2$ (purple) following the keep-in-core method.}
    \label{fig:OMPM}
\end{figure}

We consider the execution phases to be preemptive and preemption is realized using the \emph{keep-in-core} preemption method \cite{AER_preempt}, where the data of the preempted task is kept in the core's local memory during preemption until it resumes execution (see Figure \ref{fig:OMPM}). 
At the time of preemption, the execution of the preempted task is paused, and the read phase of the preempting task is started.
The data and code of the preempting task are stored in local memory above the data and code of the preempted task (like a stack).
Hence it is crucial that data of all preempted tasks and the running task fit into local memory simultaneously.
Once all 3 phases of the preempting task are completed, the preempted task's execution phase can resume. 



We say a task set is \emph{schedulable} if the worst-case response time (WCRT) of each task is below its deadline and \emph{memory-feasible} if the worst-case memory usage of tasks assigned to all cores is below their local memory size.

\section{Background on the Response Time Analyses}
\label{sec:BackAnalysis}


In this paper, we present two WCRT analyses for preemption threshold scheduling for 3-phase tasks under global task-priority memory access scheduling (see Section \ref{sec:WCRT_taskPrio}) and under global core-priority memory access scheduling (see Section \ref{sec:WCRT_corePrio}).
This Section provides the necessary background context for the presented analyses.
Our response time analyses for 3-phase tasks with preemption thresholds utilizes techniques of existing response time analyses for limited preemptive scheduling using preemption thresholds for conventional tasks \cite{Buttazzo2024, wang1999scheduling, regehr2002scheduling} and on existing response time analysis for fully preemptive E-phases \cite{arora2022analyzing} and non-preemptive E-phases \cite{schwaricke2020fixed} of 3-phase tasks.






In fixed-priority limited preemptive scheduling, the WCRT of a task $\tau_i$ executed on the local core $\pi_l$ is observed during the longest level-i active period $L_{i,l}$ \cite{bril2007worst}. 
In non-preemptive and limited preemptive scheduling, a high-priority task released during a non-preemptive section of a low-priority task can be pushed ahead to successive jobs, causing the successive jobs to experience higher interference.
This anomaly is called the \emph{self-pushing phenomenon} \cite{Buttazzo2024}.
To consider possible effects of the \emph{self-pushing phenomenon}, all jobs in the level-i active period must be analyzed when calculating the WCRT \cite{Buttazzo2024, bril2007worst, regehr2002scheduling}. This is due to the fact that the worst-case response time does not necessarily occur during the first instance of a task.
The level-i active period is based on the pending load. 



 \begin{definition}[Level-i pending load $\mathit{W_i^p(t)}$ \cite{bril2007worst}] 
 \label{def:pending}
 Level-i pending load $\mathit{W_i^p(t)}$ at time $t$ is the amount of processing that still needs to be performed at time $t$ due to jobs with priority higher than or equal to $\mathit{P_i}$ released strictly before $t$.
 \end{definition}
 
 \begin{definition}[Level-i active period \cite{bril2007worst}]
 \label{def:activeP}
A level-i active period $L_i$ is an interval  $\mathit{[a, b)}$ such that the level-i \emph{pending} workload $\mathit{W^P_i(t)}$ is positive for all $\mathit{t \in (a, b)}$ and zero in $a$ and $b$.
 \end{definition}

In order to calculate the WCRT of a 3-phase task $\mathit{\tau_i}$ assigned to core $\pi_l$, the level-i active period $L_{i,l}$ must be calculated. 
 There exist four types of delays on a task, $\tau_i$, that must be accounted for due to other tasks executing on the local core, $\pi_l$ (the same core as the task under analysis is assigned to) and on remote cores, $\pi_r$ (other cores) as identified in \cite{arora2022analyzing, schwaricke2020fixed}.
To bound $\mathit{L_{i,l}}$, these four delays must be bounded for a task $\tau_i$.

 \begin{enumerate}
     \item Intra-core interference ($I_i$) - maximum interference that can be caused on a task $\tau_i$ by all higher or equal priority tasks ($hep_{i,l}$) released on the local core $\pi_l$.
     \item Intra-core blocking ($B_i$) - maximum blocking that can be caused on a task $\tau_i$ due to lower priority tasks ($lp_{i,l}$) that execute on the local core $\pi_l$.
     \item Inter-core interference ($I_i^{mem}$) - The maximum memory interference that can be suffered by a task $\tau_i$ from all higher or equal priority tasks ($hep_{i,r}$) on all remote cores $\pi_r$.
     \item Inter-core blocking ($B_i^{mem}$) - maximum memory blocking that can be suffered by a task $\tau_i$ by all lower priority tasks ($lp_{i,r}$) that execute on all remote cores $\pi_r$.
 \end{enumerate}

Figure \ref{fig:active_period_example} shows an example active period where \circled{1}-\circled{7} represents interference and blocking from the factors above.
Sections \ref{sec:WCRT_taskPrio} and \ref{sec:WCRT_corePrio} present how to bound the above 4 factors for our model, as well as the main steps to calculate WCRT, i.e., calculation of the active period, $\mathit{L_{i,l}}$, using the above 4 factors, calculating the response time for each job in the active period to find the worst-case response time among them.

\begin{figure}
    \centering
    \includegraphics[width=\linewidth]{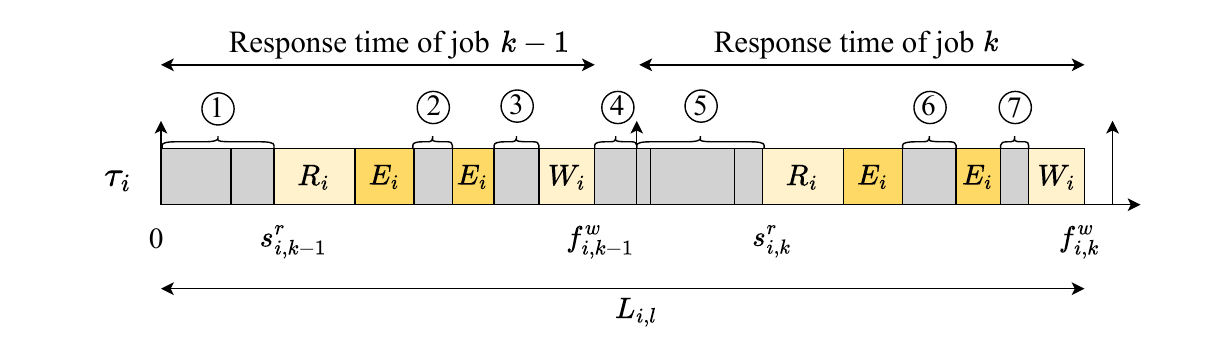}
    \caption{Illustration of a possible active period of $\mathit{\tau_i}$. The active period has 2 jobs of $\mathit{\tau_i}$. Gray colored durations represent possible $I_i$, $B_i$, $I^{mem}_i$, and $B^{mem}_i$ delays.}
    \label{fig:active_period_example}
\end{figure}



\subsection{Task Categorization}
\label{sec:taskCategorization}

Having a nominal priority as well as a preemption threshold makes bounding $I_i$, $B_i$, $I^{mem}_i$, and $B^{mem}_i$ terms challenging.
To simplify the problem, we categorize other tasks in the task set into 6 categories depending on their nominal priorities and preemption thresholds with respect to those of $\tau_i$ (see Figure \ref{fig:categories}). 

\begin{itemize}
    \item [\circled{A}] - Tasks that have a nominal priority and preemption thresholds below the nominal priority of $\mathit{\tau_i}$, i.e., $\mathit{P_j \leq \theta_j < P_i \leq \theta_i}$. 
    
    \item [\circled{B}] - Tasks that have a nominal priority lower than the nominal priority of $\mathit{\tau_i}$ and a preemption threshold higher than or equal to the nominal priority but lower than the preemption threshold of $\mathit{\tau_i}$, i.e., $\mathit{P_j< P_i \leq \theta_j < \theta_i}$.
    
    \item [\circled{C}] - Tasks that have a nominal priority and preemption threshold higher than or equal to the nominal priority of $\mathit{\tau_i}$ but are lower than or equal to the preemption threshold of $\mathit{\tau_i}$, i.e., $\mathit{P_i \leq P_j \leq \theta_j \leq \theta_i}$.
    
    \item [\circled{D}] - Tasks that have a nominal priority higher than or equal to the nominal priority but lower than or equal to the preemption threshold of $\mathit{\tau_i}$ and have a preemption threshold higher than $\mathit{\tau_i}$'s preemption threshold, i.e., $\mathit{P_i \leq P_j \leq \theta_i < \theta_j}$.
    
    \item [\circled{E}] - Tasks that have a nominal priority and a preemption threshold higher than $\mathit{\tau_i}$'s preemption threshold, i.e., $\mathit{P_i \leq \theta_i < P_j \leq \theta_j}$.
    
    \item [\circled{F}] - Tasks that have a nominal priority lower than $\mathit{\tau_i}$'s priority but have a higher or equal preemption threshold compared to $\mathit{\tau_i}$'s preemption threshold, i.e., $\mathit{P_j < P_i \leq \theta_i \leq \theta_j}$.
\end{itemize}

\begin{figure}[b]
    \centering
    \includegraphics[width=0.8\linewidth]{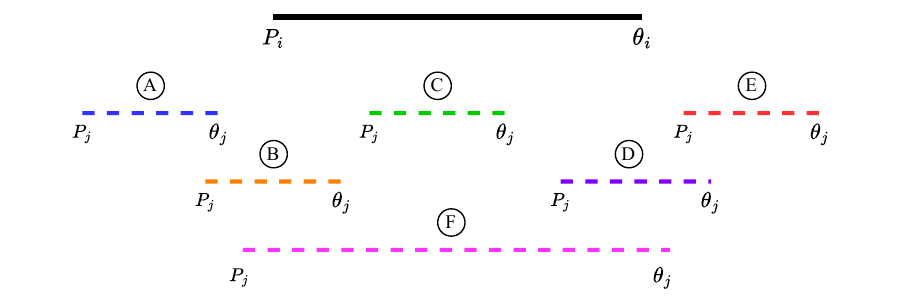}
    \caption{Categories of tasks that could interfere or block $\mathit{\tau_i}$}
    \label{fig:categories}
\end{figure}

The following rules were considered when carefully defining the above categories to account for equal priorities and preemption thresholds.
1) All tasks must fall under a category, 2) no task should fall under 2 categories as that task would be considered for interference and blocking twice otherwise, 3) Groups \circled{A}, \circled{B} and \circled{F} should not have tasks that have a nominal priority equal to $\tau_i$'s nominal priority as tasks that have an equal nominal priority are considered not to cause blocking as the lower priority tasks but to cause interference. Thus, they should fall under groups \circled{C}, \circled{D}, or \circled{E}, and 4) Group \circled{E} should not have tasks that have nominal priority equal to the preemption threshold of $\tau_i$ as such tasks can't preempt $\tau_i$ as they should in group \circled{E}.
Once the tasks are categorized, the next step is to bound $I_i$, $B_i$, $I_i^{mem}$, and $B_i^{mem}$ and calculate the WCRT for 3-phase tasks under preemption threshold scheduling.

Next sections present the WCRT analyses for preemption threshold scheduling for 3-phase tasks under global task-priority memory access scheduling (see Section \ref{sec:WCRT_taskPrio}) and under global core-priority memory access scheduling (see Section \ref{sec:WCRT_corePrio}). In literature the phased model is studied under both core-priority \cite{schwaricke2020fixed} and task-priority based \cite{arora2022analyzing} memory centric schedulers.
Therefore, we also explore both options to evaluate them under preemption threshold scheduling. 

\section{Response Time Analysis for Task-priority Based Memory-Centric Scheduling}
\label{sec:WCRT_taskPrio}

This Section presents the response time analysis for task-priority based memory-centric scheduling. 
The read and write memory phases are considered to be non-preemptive and the execution phases are preemptive.
Bounding the interference and blocking is the first step (see Section \ref{sec:IandB_taskPrio}) followed by bounding the WCRT (see Section \ref{sec:boundWCRT_taskPrio}).

\subsection{Bounding Interference and Blocking}
\label{sec:IandB_taskPrio}

\subsubsection{Bounding Intra-Core Interference}
The intra-core interference a task $\mathit{\tau_i}$ can experience is given by Lemma \ref{lem:intraI}.
\begin{lemma}
\label{lem:intraI}
   The maximum interference a task $\mathit{\tau_i}$ can experience from tasks assigned to the local core during $\mathit{L_{i,l}}$, is upper-bounded by $\mathit{I_i(L_{i,l})}$, where
  \begin{equation}
  \label{Eq:IntraI}
     I_i(L_{i,l}) = \sum_{\tau_h \in hep_{i,l}} \eta^+_h (L_{i,l}) \times C_h
  \end{equation}
\end{lemma}
\begin{proof}
    Only tasks with a higher or equal nominal priority to $\mathit{\tau_i}$'s nominal priority can cause interference to it. Thus, out of the categories \circled{A} - \circled{F}, any task $\mathit{\tau_h}$ assigned to the local core in categories \circled{C}, \circled{D} and \circled{E} can interfere with $\mathit{\tau_i}$ as they have a higher or equal nominal priority than $\mathit{\tau_i}$. All the local tasks in categories \circled{C}, \circled{D}, and \circled{E} can be grouped as higher or equal priority tasks on the local core, i.e., $\mathit{\forall \tau_h \in (\circled{C} \cup \circled{D} \cup \circled{E}) \in hep_{i,l}}$.  To upper bound the maximum intra-core interference on a task $\tau_i$, the maximum number of jobs that can be released by a task $\mathit{\tau_h}$ within the duration $\mathit{L_{i,l}}$ must be bounded, which is represented by $\mathit{\eta^+_h(t)}$ (see Section \ref{sec:boundWCRT_taskPrio}). Each time a task $\mathit{\tau_h}$ is released, it can interfere with $\mathit{\tau_i}$ for a duration of its total execution time. Thus, to upper-bound the maximum intra-core interference, we multiply the maximum number of releases by their respective total execution times.
    Tasks in categories \circled{A}, \circled{B}, and \circled{F} can't cause interference on $\mathit{\tau_i}$ as they have lower nominal priorities than $\mathit{\tau_i}$. $\square$
\end{proof}

While tasks in all three categories \circled{C}, \circled{D}, and \circled{E} can interfere with a task $\tau_i$ before it starts executing (before starting the read phase) as they have a higher or equal nominal priority than $\mathit{\tau_i}$, once the task starts to execute, only a task $\mathit{\tau_h}$ in category \circled{E} can interfere during the execution phase by preempting $\mathit{\tau_i}$ as $\mathit{\theta_i < P_h}$ in category \circled{E}. 
We make this distinction later in Section \ref{sec:boundWCRT_taskPrio}.

\subsubsection{Bounding Intra-Core Blocking}
The maximum blocking a task $\mathit{\tau_i}$ can experience from lower-priority tasks on the same core is bounded in Equation \ref{Eq:IntraB}.
\begin{lemma}
\label{lem:intraB}
    The maximum intra-core blocking that can be suffered by a task $\mathit{\tau_i}$ due to lower priority tasks assigned to the local core during the time period $\mathit{L_{i,l}}$ is upper bounded by $\mathit{B_i}$, where
\begin{equation}
 \label{Eq:IntraB}
     B_i = \max( \max_{\forall \tau_j \in \circled{A}} {C_j^r}, \max_{\forall \tau_j \in \circled{A}} {C_j^w} ,\max_{\tau_j \in \circled{B} \cup \circled{F}} C_j)
 \end{equation}
\end{lemma}
\begin{proof}
    Out of the tasks assigned to the same core as $\mathit{\tau_i}$, only the tasks that have a nominal priority lower than $\mathit{\tau_i}$'s nominal priority (i.e., lower-priority tasks) can cause intra-core blocking, i.e., tasks in categories \circled{A}, \circled{B}, and \circled{F}.
    Intra-core blocking can only occur once during $\mathit{L_{i,l}}$ by only one lower priority task on the local core. By the end of the blocking, the fixed-priority scheduler schedules $\mathit{\tau_i}$ or the highest-priority ready task. In addition, for the rest of the active period, there will be pending load from higher or equal priority tasks (see definitions \ref{def:pending} and \ref{def:activeP}); therefore, a lower priority task can not be scheduled until the end of the active period. The only possible intra-core blocking occurs when $\mathit{\tau_i}$ is released during a non-preemptive phase of a lower-priority task. The worst-case scenario is when a lower priority task starts its non-preemptive phase just before $\mathit{\tau_i}$'s release. 
    To account for the worst-case scenario, the maximum blocking that can be caused by a task in groups \circled{A}, \circled{B}, and \circled{F} is considered in Equation \ref{Eq:IntraB}. For this, we need to determine the maximum duration a lower priority task in each category can block $\mathit{\tau_i}$.
    Out of these lower-priority tasks, tasks in category \circled{A} can only block task $\mathit{\tau_i}$ for a maximum of either read or write phase time due to non-preemptive memory phases. Upon completion of the memory phase, $\mathit{\tau_i}$ can preempt a lower priority task in category \circled{A}.
    However, the tasks in \circled{B} and \circled{F} can block $\mathit{\tau_i}$ for their total execution times. If a task $\mathit{\tau_j} \in \circled{B} \cup \circled{F}$ starts to execute just before $\mathit{\tau_i}$ becomes ready, $\mathit{\tau_i}$ is not able to preempt $\mathit{\tau_j}$ even after its non-preemptive read phase is completed due to the fact that $\mathit{\tau_j}$'s priority is raised to its preemption threshold. Since $\mathit{P_i < \theta_j}$, $\mathit{\tau_i}$ can not preempt $\mathit{\tau_j}$. Thus, will experience a blocking time of $\mathit{\tau_j}$'s total execution time.
    To account for the maximum blocking that can be suffered, the maximum blocking that can be caused by a task in groups \circled{A}, \circled{B}, and \circled{F} is calculated in equation \ref{Eq:IntraB}. $\square$
\end{proof}

\subsubsection{Bounding Inter-core memory Blocking}
\label{sec:interB}

\begin{figure}[b]
    \centering
    \includegraphics[width=\linewidth]{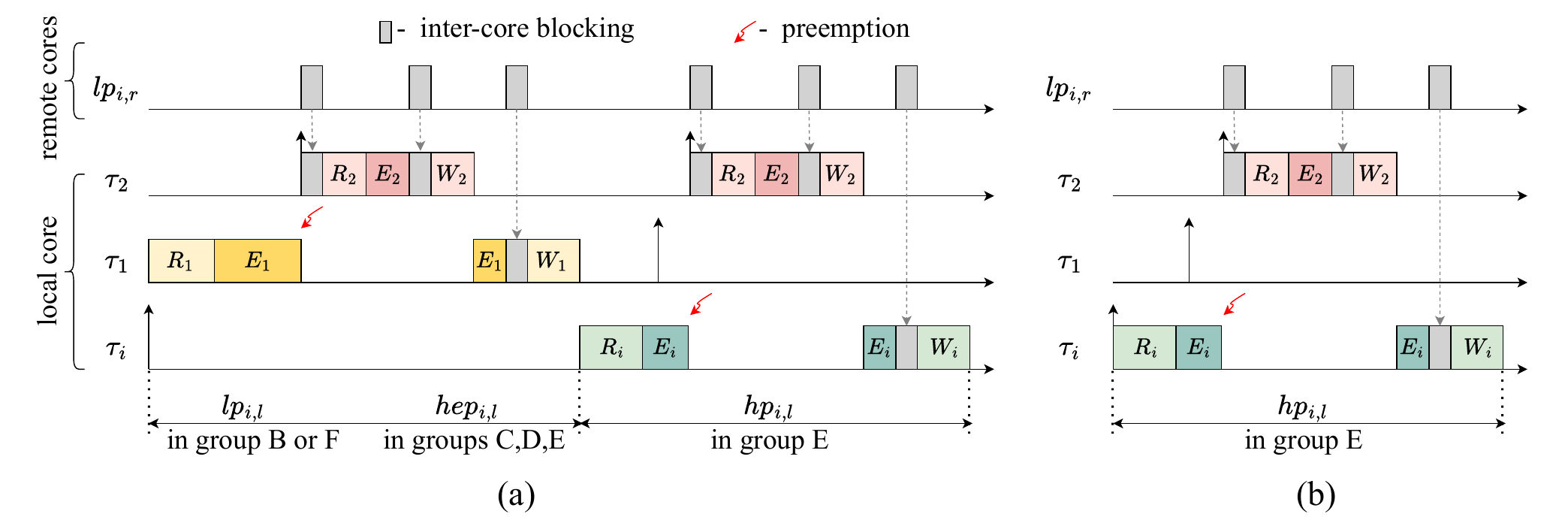}
    \caption{An example of direct and indirect inter-core blockings caused on $\tau_i$. $\tau_1$ is a task in group B or F. $\tau_2$ is in group E. Scenario 1 on the right, scenario 2 on the left.}
    \label{fig:InterB}
\end{figure}

This section bounds the inter-core blocking on a task $\mathit{\tau_i}$. 
Three steps are followed to bound $\mathit{B^{mem}_i}$.

\begin{enumerate}
    \item Identifying all the remote memory phases that can \emph{cause} inter-core memory blocking.
    \item Identifying all the local memory phases that can \emph{suffer} from inter-core memory blocking.
    \item Upper bounding the maximum inter-core memory blocking during the active period, $\mathit{L_{i,l}}$.
\end{enumerate}

 Analyzing from both the perspectives of the local core and remote cores to reduce pessimism has been previously used by related work \cite{arora2022analyzing, schwaricke2020fixed}. We follow similar steps as the maximum number of inter-core memory blockings that can be suffered by a task $\mathit{\tau_i}$ is not always the same as the maximum number of inter-core memory blockings that can be caused by lower priority tasks on remote cores.

\smallskip\noindent\textbf{1) Identifying $M$:}
This step identifies all the memory phases on remote cores that can cause inter-core memory blocking on $\mathit{\tau_i}$ and adds the execution times of these memory phases to a list, $M$.

\begin{lemma}
\label{lem:mu}
   The execution times of all the memory phases of lower priority tasks running on remote cores that can cause inter-core memory blocking on a task $\mathit{\tau_i}$ executing on the local core during $\mathit{L_{i,l}}$ is identified by the list $\mathit{M}$, where
    \begin{equation}
    \label{EQ:caused_blocking}
    M=\{ C_{q,k}^r, C_{q,k}^w | 0 \leq k \leq \eta^+_q (L_{i,l} + C_q^{rt} - C_q), \forall r \in [1,m],r \neq l, \forall \tau_q \in lp_{i,r} \}
    \end{equation}

\end{lemma}

\begin{proof}
    Only lower priority tasks on remote cores can cause inter-core memory blocking to $\mathit{\tau_i}$, i.e., $\mathit{\forall \tau_j \in (\circled{A}, \circled{B}, \circled{F}) \in lp_{i,r}}$. To determine all the memory phases that can cause blocking, the maximum number of jobs each lower-priority task can execute on any remote core during the active period must be determined. 
    In addition to accounting for all lower-priority jobs released on remote cores during the active period, $\mathit{L_{i,l}}$, all possible carry-in jobs that are not released during $\mathit{L_{i,l}}$, but might execute during $\mathit{L_{i,l}}$ must be accounted for. Such cases can be treated as inherited release jitter, as the start of a memory phase on a remote core may be delayed due to the activities on that core that we do not directly consider when bounding $\mathit{B^{mem}_i}$. 
    The release jitter of a task on a remote core is equal to its latest start time, as used by \cite{schwaricke2020fixed}. The latest start time of a task $\tau_j$ is upper bounded by $\mathit{C_j^{rt} - C_j}$ for remote blocking, according to \cite{chen2019many}, where $\mathit{C_j^{rt}}$ represents the WCRT of $\tau_j$.
    Thus, accounting for all jobs released in the duration $\mathit{L_{i,l} + C_j^{rt} - C_j}$ accounts for all possible jobs of a lower-priority remote task that can cause inter-core blocking on $\tau_i$. 
    Accordingly, the maximum number of jobs that can cause inter-core memory blocking by remote tasks is bounded by $\mathit{\sum_{r=1, r \neq l}^{m}  \sum_{\tau_j \in lp_{i,r}} \eta^+_j (L_{i,l} + C_j^{rt} - C_j)}$. 
    Each such job that is identified can cause memory blockings during its 2 memory phases for the respective durations of its read and write phases, i.e., $\mathit{(C_j^r + C_j^w)}$.
    Thus, the respective $C_j^r$ and $C_j^w$ for each identified job is added to $M$.
    $\square$ 
\end{proof}
\begin{equation}
      M = \{{C^{r/w}_1, C^{r/w}_2,..., C^{r/w}_V | C^{r/w}_x \geq C^{r/w}_{x+1}}\}
     \label{Eq:M}
\end{equation}

$M$ identified in Lemma \ref{lem:mu} is sorted in decreasing order (see Equation \ref{Eq:M}). 
The tasks are analyzed starting from the lowest priority task as the WCRT of lower priority remote tasks are used in Lemma \ref{lem:mu}.

\smallskip\noindent\textbf{2) Bounding $\mathit{\Phi_i(L_{i,l})}$:} The maximum number of memory phases that can be suffered is given in Lemma \ref{lem:phi}. A few steps are followed to calculate $\mathit{\Phi_i(L_{i,l})}$.

\begin{lemma}
\label{lem:phi}
    The maximum number of memory phases on the local core that can suffer inter-core memory blocking during $\mathit{L_{i,l}}$, resulting in inter-core memory blocking on a task $\mathit{\tau_i}$ executing on the local core upper bounded by $\mathit{\Phi_i(L_{i,l})}$, where
    \begin{equation}
    \label{Eq:suffered}
      \Phi_i(L_{i,l}) = \eta^+_i (L_{i,l}) \times 2 + \sum_{\tau_h \in hep_{i,l}} \eta^+_h (L_{i,l}) \times 2
    \end{equation}
\end{lemma}

\begin{proof}
    Every time an E-phase is executed on the local core, there is a possibility for a memory phase to be scheduled on a remote core. It is only during the E-phases that the local
    core can not schedule any memory phases during the level-i active window, and instead, a lower priority memory phase on a remote core can be scheduled on the shared interconnect. Thus, every time a memory phase is scheduled on the local core after an E-phase, e.g., at time $t$, there is a possibility that it is blocked by a lower-priority memory phase on a remote core that was scheduled just before it was ready, i.e., at $t-\epsilon$ where $\epsilon $ is a very small time duration, (see Figure \ref{fig:InterB}). Therefore, in order to determine the maximum number of memory phases the local core can suffer inter-core memory blocking, the maximum number of times a memory phase may execute after an E-phase on the local core during $\mathit{L_{i,l}}$ must be determined. This includes possible inter-core memory blocking before $\mathit{\tau_i}$'s own memory phases (direct) and before memory phases of other blocking and interfering tasks on the local core during $\mathit{L_{i,l}}$ (indirect).
    There are two and only two possible scenarios. 1) The read phase of $\mathit{\tau_i}$'s first job within $\mathit{L_{i,l}}$ is scheduled after an E-phase, and 2) The read phase of $\mathit{\tau_i}$'s first job within $\mathit{L_{i,l}}$ is scheduled after a memory phase. 
    
    Scenario 1 only occurs if $\mathit{\tau_i}$ is released during an E-phase of a task that it can preempt, i.e., category \circled{A}. In this case, there can be direct inter-core memory blocking before both $\mathit{\tau_i}$'s memory phases. Its read phase starts after an E-phase of the lower priority task it preempted, and $\mathit{\tau_i}$'s write phase is scheduled after its own E-phase.
    In addition, $\mathit{\tau_i}$ can experience inter-core memory blocking indirectly due to higher priority tasks during its E-phase and higher or equal priority tasks after, depending on how many instances $\mathit{\tau_i}$ has within $\mathit{L_{i,l}}$. To elaborate more, during its execution phase, $\mathit{\tau_i}$ can experience interference (through preemption) from tasks that have a priority higher than its preemption threshold (\circled{E}). If the active period continues to more than one instance of $\mathit{\tau_i}$, there can be inter-core memory blocking due to the higher or equal priority tasks (\circled{C}, \circled{D}, and \circled{E}) that may execute between executions of $\mathit{\tau_i}$'s instances within the active period. A maximum of 2 indirect memory blocking can occur due to each memory phase of $hep_{i,l}$.
    However, in scenario 1, there can't be any indirect inter-core memory blocking due to lower-priority tasks, i.e., categories \circled{A}, \circled{B}, and \circled{F}. This is because there can only be an execution of such a task \emph{once} during $\mathit{L_{i,l}}$ (see Lemma \ref{lem:intraB}), and it does not occur when $\mathit{\tau_i}$ is scheduled by preemption. Thus, the maximum number of inter-core memory blocking that can be suffered in scenario 1 is $\sum_{\tau_h \in hep_{i,l}} \eta^+_h (L_{i,l}) \times 2 + \eta^+_i (L_{i,l}) \times 2$.

    If not scenario 1, then the read phase of $\mathit{\tau_i}$ is scheduled after a memory phase on the local core, which falls under scenario 2. In this case, there can only be 1 direct inter-core memory blocking on $\mathit{\tau_i}$'s first job within $\mathit{L_{i,l}}$. That is, before scheduling $\mathit{\tau_i}$'s write phase after its own E-phase, as a lower-priority memory phase can't be scheduled on a remote core before its read phase. However, there can be two direct memory blockings on $\mathit{\tau_i}$'s latter jobs.
    There can be indirect inter-core memory blocking when scheduling memory phases of higher or equal priority tasks after E-phases, i.e., categories \circled{C}, \circled{D}, and \circled{E}, and a memory phase of at most one lower priority task in either category \circled{B} or \circled{F} (see Lemma \ref{lem:intraB}) after its E-phase. 
    The number of intra-core interfering and blocking tasks are presented in Lemmas \ref{lem:intraI} and \ref{lem:intraB}, respectively. 
    For each higher or equal priority task that is released during $\mathit{L_{i,l}}$, there can be 2 inter-core memory blocking occurrences. However, there can only be one inter-core memory blocking suffered due to the possible intra-core blocking task. That is when it schedules its write phase. Possible inter-core memory blocking on its read phase is outside $\mathit{L_{i,l}}$ as for this situation to occur, the read phase of the intra-core blocking task must have started at least $\epsilon$ time before the release of $\mathit{\tau_i}$.
    There can't be any indirect inter-core memory blocking due to lower priority tasks in category \circled{A} as an E-phase, and neither of their memory phases is started during $\mathit{L_{i,l}}$ (see Lemma \ref{lem:intraB}).
    Thus, the maximum number of inter-core memory blocking that can be suffered in scenario 2 is $\sum_{\tau_h \in hep_{i,l}} \eta^+_h (L_{i,l}) \times 2 + 1 + \eta^+_i (L_{i,l}) \times 2 - 1$ (-1 as there can only be one direct inter-core memory blocking on $\mathit{\tau_i}$'s first job and +1 from indirect inter-core memory blocking on one lower priority task).
    Thus, in either scenario, $\Phi_i(L_{i,l})$ is the same and is what is presented in Lemma \ref{lem:phi}.
    $\square$
    
\end{proof}

\smallskip\noindent\textbf{3) Bounding $\mathit{B^{mem}_i(L_{i,l})}$:}
This step bounds the maximum inter-core memory blocking on $\mathit{\tau_i}$. 

A maximum of one remote memory phase can cause blocking to each local memory phase identified in Lemma \ref{lem:phi}, as the waiting local memory phase would be scheduled on the interconnect before a second lower-priority remote memory phase. However, this is with respect to their own task priorities, which may differ from the perspective of the analyzed task $\mathit{\tau_i}$. In addition, each memory phase can suffer from inter-core interference. Consider the local memory phase of the lower priority task in either category \circled{B} or \circled{F} that could be scheduled within the active period (e.g., $W_1$ in Figure \ref{fig:InterB}(a)). Remote tasks that cause inter-core interference on such a memory phase within $\mathit{L_{i,l}}$ can be inter-core blocking from the perspective of $\mathit{\tau_i}$ (e.g., when $\mathit{ \tau_j \in lp_{i,r} ; P_j \geq P_1}$ in Figure \ref{fig:InterB}).
Therefore, such cases must also be accounted for in the analysis when bounding inter-core blocking.

Let $\delta_i$ be the priority of the lowest priority task out of tasks in $\circled{B}_i$ and $\circled{F}_i$ (if $\circled{B}_i$ and $\circled{F}_i$ are not empty). To account for the worst-case, we consider the case where the lowest priority ($\delta_i$) task among tasks in $\circled{B}_i$ and $\circled{F}_i$ causes intra-core blocking to $\mathit{\tau_i}$ (represented by the first term in Equation \ref{Eq:B_mem}). Note that after considering the maximum possible inter-core blocking caused by a task,$\tau_l$, its memory phases are removed from $M$ to avoid accounting for them again in the second term of Equation \ref{Eq:B_mem}.
The second term in the Equation \ref{Eq:B_mem} represents the maximum possible inter-core memory blocking suffered by each local memory phase in $\mathit{\Phi_i(L_{i,l})}$ (at most one lower-priority inter-core memory phase per local memory phase). Which is bounded by summing the largest remote lower-priority memory phases in $M$ up to $\Phi_i\mathit{(L_{i,j})}$ where $\mathit{C^{r/w}_x} \in M$. The number of memory phases in $M$ can be less than $\mathit{\Phi_i(L_{i,l})}$, i.e., $len(M) < \Phi_i(L_{i,l})$, in which case we sum up the memory phases in $M$. In other words, the minimum is taken between $\mathit{\Phi_i(L_{i,l})}$ and $len(M)$ as inter-core memory can't be suffered by $\mathit{\tau_i}$ more than what is caused by remote tasks and vice versa. 

The maximum inter-core memory blocking that can be suffered by a task $\mathit{\tau_i}$ executing on the local core during $\mathit{L_{i,l}}$ is upper bounded by $\mathit{B^{mem}_i(L_{i,l})}$ (Equation \ref{Eq:B_mem}). Note that the first term of the Equation \ref{Eq:B_mem} provides 0 if $\circled{B}_i$ and $\circled{F}_i$ are empty.

    \begin{equation}
    \label{Eq:B_mem}
      B^{mem}_i(L_{i,l}) =   \sum_{\delta_i \leq P_l < P_i} \eta^+_l (L_{i,l} + C_l^{rt} - C_l) \times (C^r_l + C^w_l) + \sum_{x=1}^{ min({\Phi_i(L_{i,l}), len(M))}}  C^{r/w}_x
    \end{equation}

\subsubsection{Bounding Inter-Core Interference}

This section bounds the maximum inter-core interference that can be caused on a task $\mathit{\tau_i}$ (see Lemma \ref{lem:interI}). 
\begin{lemma}
\label{lem:interI}
    The maximum inter-core interference that can be caused by higher or equal priority tasks $\mathit{\tau_u}$ running on remote cores on a task $\mathit{\tau_i}$ during $\mathit{L_{i,l}}$ is upper bounded by $\mathit{I_i^{mem}(L_{i,l})}$ where
  \begin{equation}
  \label{Eq:InterI}
 \begin{split}
     I_i^{mem}(L_{i,l}) = \sum_{r=1, r \neq l}^{m} \sum_{\tau_u \in hep_{i,r}} (\eta^+_u (L_{i,l}) + 1) \times (C_u^r + C_u^w)
 \end{split}
 \end{equation}
    
\end{lemma}

\begin{proof}
    Only higher or equal priority tasks that are assigned to remote cores can cause interference to $\mathit{\tau_i}$ when scheduling its memory phases (read and write) on the shared interconnect, i.e., $\mathit{\forall \tau_j \in (\circled{C} \cup \circled{D} \cup \circled{E}) \in hep_{i,r}}$. Thus, to get the maximum memory interference caused, the maximum number of times each higher or equal priority task is released during $\mathit{L_{i,l}}$ is multiplied by the addition of their respective read and write phase times. 
    Similar to Lemma \ref{lem:mu}, we must also account for the carry-in jobs of a task $\tau_j$ that are delayed due to interference and blocking from tasks executing in their respective core.
    We greedily add +1 to each task, as at most one job can be carried in by each higher-priority remote task considered. 
    Once the maximum number of jobs that can cause inter-core interference is identified for each task, it is multiplied by its respective read phase time and the write phase time, as each memory phase can cause interference.
    $\square$
\end{proof}

The carry-in jobs of interfering remote tasks can also be treated as release jitter, as in Lemma \ref{lem:mu}. Due to the use of worst-case response time in the calculation, the carry-in jobs as release jitter can only be used in either inter-core interference or inter-core blocking.
We chose to model carry-in jobs as release jitter when bounding inter-core blocking, as $\mathit{B_i^{mem}(L_{i,l})}$ is more sensitive to pessimism. 
In addition, the analysis does not analyze from the perspectives of both the local core and the remote cores to compare the inter-core interference that can be caused with the interference that can be suffered by the local core similar to in Section \ref{sec:interB} when bounding the inter-core blocking. This is because getting a tighter bound on the amount of inter-core interference the local memory phases can suffer during the active period than what can be caused by the remote memory phases.
Furthermore, Equation \ref{Eq:InterI} bounds the inter-core blocking on a higher priority local task during the active period that is inter-core interference from the perspective of $\tau_i$, the task under analysis.
For example, $\tau_2$ in Figure \ref{fig:InterB} is a local higher-priority task, i.e., $\tau_2 \in \mathit{hep_{i,l}}$. There can be upto one remote memory phase causing inter-core memory blocking on $\tau_2$ during $\mathit{L_{i,l}}$. This memory phase can be a higher-priority remote memory phase compared to $\tau_i$, while being a lower-priority remote task from $\tau_2$'s perspective, which should be accounted as inter-core interference. Equation \ref{Eq:InterI} accounts for such memory phases as well.

\subsection{Bounding the WCRT}
\label{sec:boundWCRT_taskPrio}

This section presents how to calculate the WCRT using the inter- and intra-core blocking and interference factors discussed in the previous section. 
Several steps are followed. 
The first step is to calculate the \emph{longest} level-i active period, $\mathit{L_{i,l}}$ using the bounded $I_i$, $B_i$, $I^{mem}_i$, and $B^{mem}_i$ values that contribute to it (Equation \ref{Eq:calc_busyPeriod}).
For a sporadic task $\tau_j$, the maximum number of releases during a time duration $\mathit{L_{i,l}}$, $\mathit{\eta^+_j(L_{i,l})}$ is given by $\mathit{\left\lceil{\frac{L_{i,l}}{T_j}}\right\rceil}$ \cite{richter2003scheduling}.
Thus, in Equation \ref{Eq:calc_busyPeriod} it is used when bounding $I_i$, $B_i$, $I^{mem}_i$, and $B^{mem}_i$ using their respective equations in Section \ref{sec:IandB_taskPrio}.
When calculating $\mathit{I_i(L_{i,l})}$ and $\Phi_i$ for $\mathit{B_i^{mem}(L_{i,l})}$, higher or equal priority tasks in categories \circled{C}, \circled{D} and \circled{E} must be considered, i.e., $\mathit{\forall \tau_j \in (\circled{C}, \circled{D}, \circled{E}) \in hep_{i,l}}$ to account for the longest active period. 
$\mathit{L_{i,l}}$ is given by the first positive fixed-point solution of the recurrent relation in Equation \ref{Eq:calc_busyPeriod}.

\begin{equation}
\label{Eq:calc_busyPeriod}
\begin{split}
     L_{i,l} = I_i(L_{i,l}) + B_i + \eta^+_i(L_{i,l}) \times C_i + I_i^{mem}(L_{i,l}) + B_i^{mem}(L_{i,l}) 
\end{split}     
\end{equation}

Then the maximum number of jobs released within its active period, $K_i$, for task $\mathit{\tau_i}$ is calculated using Equation \ref{Eq:cal_no_of_jobs}.
\begin{equation}
\label{Eq:cal_no_of_jobs}
     K_i = \eta^+_i(L_{i,l}) =  \left\lceil{\frac{L_{i,l}}{T_i}}\right\rceil
 \end{equation}

In the example scenario illustrated in Figure \ref{fig:active_period_example}, $K_i = 2$. 
The response time must be calculated for each job, $\tau_{i,k}$, within the active period (see Section \ref{sec:BackAnalysis}). 

Groups of tasks that can interfere with a task $\tau_i$ are different before and after the priority is raised to its preemption threshold. This fact can be considered when calculating the response time of each job, as also done with conventional tasks with preemption thresholds \cite{wang2016minimizing, regehr2002scheduling, wang1999scheduling, Buttazzo2024}, which results in a more tailored analysis. We follow the same basic steps followed by these works proposed for conventional tasks on single-core systems and adapt them for our 3-phase tasks on multi-core systems. Three main steps are followed, 1) calculating the start time of each job within $\mathit{L_{i,l}}$, 2) calculating the end time of each job in $\mathit{L_{i,l}}$ and 3) calculating the response times.

\subsubsection{Calculating the Start Time}
In the first step, the start time of each job, $\mathit{s^r_{i,k}}$, within the active period is calculated, which is given by the first positive fixed-point solution of the recurrent relation in Equation \ref{Eq:cal_start_read_phase}. 
$\mathit{s^r_{i,k}}$ is calculated considering the maximum amount of blocking and interference that can be caused on the $k^{th}$ job before the start of its read phase for the interval (0, $\mathit{s^r_{i,k}}$) as well as execution times for the preceding $k-1$ jobs of $\mathit{\tau_i}$, same as in \cite{Buttazzo2024}. Thus, inter- and intra-core blocking and interference during (0, $\mathit{s^r_{i,k}}$) is added to the execution times for the preceding $k-1$ jobs, $(k-1) \times C_i$.
For example, when calculating the start time of the second job in Figure \ref{fig:active_period_example}, this would be \circled{1}-\circled{5} added to $C_i$.  
When calculating $\mathit{I_i(s^r_{i,k})}$ and $\Phi_i$ for $\mathit{B_i^{mem}(s^r_{i,k})}$, higher or equal priority tasks in categories \circled{C}, \circled{D} and \circled{E} must be considered, i.e., $\mathit{\forall \tau_j \in (\circled{C}, \circled{D}, \circled{E}) \in hep_{i,l}}$ as all these 3 groups can interfere before the read phase start time.
In limited preemptive scheduling, execution of interfering tasks from the previous instances can push the start time of the current job (see Section \ref{sec:BackAnalysis}). Thus, the start times of each job can be different. To identify and account for all such scenarios, $\mathit{\eta^+_h(s^r_{i,k})} = \left\lfloor{\frac{s^r_{i,k}}{T_h}}\right\rfloor + 1$ is used in Equation \ref{Eq:cal_start_read_phase} same as previous work \cite{Buttazzo2024}.

\begin{equation}
\label{Eq:cal_start_read_phase}
\begin{split}
     s^r_{i,k} = I_i(s^r_{i,k}) + B_i + I_i^{mem}(s^r_{i,k}) + B_i^{mem}(s^r_{i,k}) + (k-1) \times C_i   
\end{split}     
\end{equation}

\subsubsection{Calculating the End Time}
The second step is to calculate the write phase finish time, $f^w_{i,k}$ of each job $\tau_{i,k}$ within the active period using the recurrent Equation \ref{Eq:cal_end_write_phase}.  
We will use the example in Figure \ref{fig:active_period_example} to better explain the steps.
$f^w_{i,k}$ can be calculated by adding the $k^{th}$ job's own execution time, $C_i$ and the inter-core blocking and inter- and intra-core interference during ($\mathit{s^r_{i,k}, f^w_{i,k}}$), that is \circled{6} and \circled{7} in Figure \ref{fig:active_period_example}, to its read phase start time $\mathit{s^r_{i,k}}$ (see Equation \ref{Eq:cal_end_write_phase}), similar to previous work \cite{Buttazzo2024}. Meaning that for the second job in the example in Figure \ref{fig:active_period_example}, $f^w_{i,k} = s^r_{i,k} + \circled{6} + \circled{7} + C_i$.
 To calculate the inter-core blocking and inter- and intra-core interference between $\mathit{s^r_{i,k}}$ and, $\mathit{f^w_{i,k}}$ (\circled{6}, \circled{7} in Figure \ref{fig:active_period_example}), first inter-core blocking and inter- and intra-core interference is calculated for the duration (0, $\mathit{f^w_{i,k}}$) (\circled{i} in the list below). Later, inter-core blocking and inter- and intra-core interference for the duration (0, $\mathit{s^r_{i,k}}$) (\circled{ii} in the list below) is deducted from it.
After $\mathit{s^r_{i,k}}$, only the tasks with a nominal priority higher than $\mathit{\tau_i}$'s preemption threshold can cause interference by preempting it.
Thus in Equation \ref{Eq:cal_end_write_phase}, when calculating $\mathit{I_i(f^w_{i,k})}$, $\mathit{I_i(s^r_{i,k})}$, and $\Phi_i$ for $\mathit{B_i^{mem}(f^w_{i,k})}$ and $\mathit{B_i^{mem}(s^r_{i,k})}$, higher or equal priority tasks \emph{only} in category \circled{E} must be considered, i.e., $\mathit{\forall \tau_j \in \circled{E} \in hep_{i,l}}$. 
These steps are in line with existing analyses for preemption thresholds for conventional tasks on single-core systems \cite{Buttazzo2024}.


\begin{itemize}
    \item [\circled{i}] - $\mathit{\eta^+_h(f^w_{i,k})} = \left\lceil{\frac{f^w_{i,k}}{T_h}}\right\rceil$ when bounding $I_i(f^w_{i,k})$, $I_i^{mem}(f^w_{i,k})$, and $B_i^{mem}(f^w_{i,k})$ as self-pushing is already accounted when calculating $\mathit{s^r_{i,k}}$.

    \item [\circled{ii}] - $\mathit{\eta^+_h(s^r_{i,k})} = \left\lfloor{\frac{s^r_{i,k}}{T_h}}\right\rfloor + 1$ when bounding $I_i(s^r_{i,k})$, $I_i^{mem}(s^r_{i,k})$, and $B_i^{mem}(s^r_{i,k})$ to deduct the interference and blocking already accounted when calculating $\mathit{s^r_{i,k}}$.
\end{itemize}

Intra-core blocking, $B_i$, is not considered in Equation \ref{Eq:cal_end_write_phase}, as it can affect $\mathit{\tau_i}$ only once before $\mathit{s^r_{i,k}}$, which is considered when calculating $\mathit{s^r_{i,k}}$.
\begin{equation}
\label{Eq:cal_end_write_phase}
\begin{split}
     f^w_{i,k} = s^r_{i,k} + C_i  + I_i(f^w_{i,k}) + I_i^{mem}(f^w_{i,k}) + B_i^{mem}(f^w_{i,k}) \\ - (I_i(s^r_{i,k}) + I_i^{mem}(s^r_{i,k}) + B_i^{mem}(s^r_{i,k}))
\end{split}     
\end{equation}




\subsubsection{Calculating the Worst-Case Response Time}
Finally, the third step is to calculate the response time of each job $\tau_{i,k}$ by deducting their release time, $r_{i,k}$, from their write phase finish time.
The worst-case response time of task $\tau_i$ is the maximum out of the response times of all jobs within the active period (see Equation \ref{Eq:cal_WCRT_final}). 
\begin{equation}
\label{Eq:cal_WCRT_final}
\begin{split}
     C^{rt}_{i} = \max_{k \in 0,1,...,\left\lfloor{\frac{L_{i,l}}{T_i}}\right\rfloor} (f^w_{i,k} - r_{i,k})
\end{split}     
\end{equation}

\subsection{Discussion}
Since our analysis is for limited preemptive scheduling using preemption thresholds, it can be used to analyze 3-phase tasks scheduled using fully preemptive execution phases, as well as non-preemptive scheduling by setting preemption thresholds accordingly.
In the case of non-preemptive scheduling, the PTs for all tasks will be the maximum priority, whereas in fully-preemptive scheduling, each task's PT is equal to its priority.

\section{Response Time Analysis for Core Priority Based Memory Centric Scheduling with Preemptive Memory Phases}
\label{sec:WCRT_corePrio}

This section presents the worst-case response time analysis for a core-priority based memory-centric scheduler. The memory phases are considered to be preemptive in contrast to Section \ref{sec:WCRT_taskPrio}. Note that memory phases are non-preemptive to the tasks assigned to the same core.

\subsection{Bounding Intra-core Interference}

The intra-core interference for core-priority based memory-centric scheduler can be bounded using the same Equation \ref{Eq:IntraI} in the task-priority based memory-centric scheduler, as the memory access scheduler does not affect the interference from tasks assigned to the same core. Additionally, the memory phases are still non-preemptive among the tasks assigned to the same core.
For completeness, the Equation \ref{Eq:IntraI_core} is presented.
Refer to Lemma \ref{lem:intraI} for the complete proof.

  \begin{equation}
  \label{Eq:IntraI_core}
     I_i(L_{i,l}) = \sum_{\tau_h \in hep_{i,l}} \eta^+_h (L_{i,l}) \times C_h
  \end{equation}

\subsection{Bounding Intra-core Blocking}

The intra-core blocking for core-priority based memory-centric scheduler can be bounded using the same Equation \ref{Eq:IntraB} in the task-priority based memory-centric scheduler for the same reasons as above. 
For completeness, the Equation \ref{Eq:IntraB_core} is presented.
Refer to Lemma \ref{lem:intraB} for the complete proof.

\begin{equation}
 \label{Eq:IntraB_core}
     B_i = \max( \max_{\forall \tau_j \in \circled{A}} {C_j^r}, \max_{\forall \tau_j \in \circled{A}} {C_j^w} ,\max_{\tau_j \in \circled{B} \cup \circled{F}} C_j)
 \end{equation}

\subsection{Bounding Inter-core Interference}

When bounding the inter-core interference for the core-priority-based memory-centric scheduler, the memory access priority, $\mathit{P^{mem}_r}$, which is the same for all tasks assigned to the same core, $\pi_r$ must be considered instead of task priorities.

\begin{lemma}
    The upper bound for the inter-core interference a task $\tau_i$ on the local core $\pi_l$ can experience from memory phases of tasks running on remote cores, $\pi_r$, that have a higher memory access priority compared to $\pi_l$ during the active period, $\mathit{L_{i,l}}$ is given by $\mathit{I_i^{mem}(L_{i,l})}$ where,

\begin{equation}
  \label{Eq:InterI_core}
 \begin{split}
     I_i^{mem}(L_{i,l}) = \sum_{P^{mem}_r \geq P^{mem}_l} \sum_{\tau_u \in \pi_r} \eta^+_u (L_{i,l} + C_u^{rt} - C_u) \times (C_u^r + C_u^w)
 \end{split}
 \end{equation}
\end{lemma}

\begin{proof}
 Only tasks assigned to cores that have a higher or equal memory access priority can cause inter-core interference to $\tau_i$ assigned to the local core $\pi_l$. Similar to Lemma \ref{lem:mu}, the carry-in jobs can be bounded by modeling the start time delay as release jitter. Thus, the latest start time of a task $\tau_u$ that can cause interference is considered, i.e., $\mathit{C_u^{rt} - C_u}$. 
 Once the maximum number of jobs of each remote task that can cause memory interference to task $\tau_i$ is identified, it is multiplied by the read and write times of the respective task, as both read and write phases can cause interference. This accounts for inter-core interference that can be caused before $\tau_i$ starting its execution as well as through preemption.
 $\square$
\end{proof}

\subsection{Bounding Inter-core Blocking}

There is no inter-core blocking caused or suffered with preemptive memory phases under core-priority memory scheduling. 

\subsection{Bounding the WCRT}
\label{sec:boundWCRT_corePrio}

The steps followed when calculating the worst-case response time for a core-priority-based memory-centric scheduler are the same as in the task-priority based memory scheduler.
We provide the equations necessary to bound the WCRT of a task $\tau_i$ with a brief explanation, while the detailed explanation is in Section \ref{sec:boundWCRT_taskPrio}.

First, the longest active period, $\mathit{L_{i,l}}$ must be calculated (see Equation \ref{Eq:calc_busyPeriod_corePrio}).
\begin{equation}
\label{Eq:calc_busyPeriod_corePrio}
\begin{split}
     L_{i,l} = I_i(L_{i,l}) + B_i + \eta^+_i(L_{i,l}) \times C_i + I_i^{mem}(L_{i,l}) 
\end{split}     
\end{equation}

Then the maximum number of jobs released within its active period, $K_i$, for task $\mathit{\tau_i}$ is calculated using Equation \ref{Eq:cal_no_of_jobs_corePrio}.
\begin{equation}
\label{Eq:cal_no_of_jobs_corePrio}
     K_i = \eta^+_i(L_{i,l}) =  \left\lceil{\frac{L_{i,l}}{T_i}}\right\rceil
 \end{equation}

Next, the start time of the read phase, $\mathit{s^r_{i,k}}$ (see equation \ref{Eq:cal_start_read_phase_corePrio}) and the end time of the write phase, $\mathit{f^w_{i,k}}$ (see Equation \ref{Eq:cal_end_write_phase_corePrio}) must be calculated.

\begin{equation}
\label{Eq:cal_start_read_phase_corePrio}
\begin{split}
     s^r_{i,k} = I_i(s^r_{i,k}) + B_i + I_i^{mem}(s^r_{i,k}) + (k-1) \times C_i   
\end{split}     
\end{equation} 

\circled{i} and \circled{ii} in the list below explain how the maximum number of releases must be calculated in Equation \ref{Eq:cal_end_write_phase_corePrio}. In addition, in Equation \ref{Eq:cal_end_write_phase_corePrio}, when calculating $\mathit{I_i(f^w_{i,k})}$ and $\mathit{I_i(s^r_{i,k})}$, higher or equal priority tasks \emph{only} in category \circled{E} must be considered, i.e., $\mathit{\forall \tau_j \in \circled{E} \in hep_{i,l}}$.

\begin{itemize}
    \item [\circled{i}] - $\mathit{\eta^+_h(f^w_{i,k})} = \left\lceil{\frac{f^w_{i,k}}{T_h}}\right\rceil$ when bounding $I_i(f^w_{i,k})$ and $I_i^{mem}(f^w_{i,k})$ as self-pushing is already accounted when calculating $\mathit{s^r_{i,k}}$.

    \item [\circled{ii}] - $\mathit{\eta^+_h(s^r_{i,k})} = \left\lfloor{\frac{s^r_{i,k}}{T_h}}\right\rfloor + 1$ when bounding $I_i(s^r_{i,k})$  + $I_i^{mem}(s^r_{i,k})$ to deduct the interference and blocking already accounted when calculating $\mathit{s^r_{i,k}}$.
\end{itemize}

 \begin{equation}
\label{Eq:cal_end_write_phase_corePrio}
\begin{split}
     f^w_{i,k} = s^r_{i,k} + C_i  + I_i(f^w_{i,k}) + I_i^{mem}(f^w_{i,k}) - (I_i(s^r_{i,k}) + I_i^{mem}(s^r_{i,k}))
\end{split}     
\end{equation}

The final step is to calculate the worst-case response time, $\mathit{C^{rt}_{i}}$ by taking the maximum response time out of the response times of all the jobs within $\mathit{L_{i,l}}$ (see Equation \ref{Eq:cal_WCRT_final_corePrio}).

\begin{equation}
\label{Eq:cal_WCRT_final_corePrio}
\begin{split}
     C^{rt}_{i} = \max_{k \in 0,1,...,\left\lfloor{\frac{L_{i,l}}{T_i}}\right\rfloor} (f^w_{i,k} - r_{i,k})
\end{split}     
\end{equation}

\subsection{Discussion}

Similar to the worst-case response time analysis for the task-priority-based memory-centric scheduler in Section \ref{sec:WCRT_taskPrio}, the analysis for the core-priority-based memory-centric scheduler can also be used to analyze not only limited preemptive scheduling but also non-preemptive and fully-preemptive scheduling by assigning suitable preemption thresholds. 

For 3-phase tasks under preemption threshold scheduling, response time analysis becomes complicated under task-priority based memory-centric scheduling, especially to bound the inter-core memory blocking and interference. Local memory phases being non-preemptive, inter-core interference and blocking causing remote tasks being different from the perspectives of different tasks on the local core, and preemption threshold of an executing local task being higher than the priority of the task under analysis, but the nominal priority of the executing task, which determines memory access, being lower, are the main reasons for this complication. 
These complications can be eliminated by a combination of allowing memory phases to be preemptive by remote tasks and using core-priority for memory access.

\section{Memory Requirement Analysis}
\label{sec:mem_requirement}


Bounding the worst-case local memory usage is important to determine the memory feasibility of a task set under preemption thresholds using the keep-in-core preemption method.
In this section, we propose a task-level memory requirement analysis for the sporadic task model considered in this paper.
When analyzing the maximum memory requirements of a task $\tau_i$, only tasks on the same core as $\tau_i$ must be considered, as tasks on other cores cannot preempt $\tau_i$.
The goal is to find the preemption chain that results in the largest cumulative memory requirement.
A so-called \textit{preemption chain} is a sequence/chain of tasks where each task is preempted by the next task in the chain, i.e., nested preemptions. Meaning that, each pair of consecutive tasks in the sequence has a preemption relationship, i.e. $\theta_k < P_{j}$.
The number of tasks in a preemption chain depends on the priority relationship between tasks on the same core \cite{memory_bound_RTNS}. We call the first task in the chain its \emph{base-task}. 

A preemption chain that starts with task $\tau_i\in\Gamma_{\pi}$ as a base-task is defined as the ordered set $PC_i=\{\tau_i, \dots, \tau_j\}$, where $\Gamma_{\pi}$ denotes the set of tasks mapped to core $\pi$.
The task at the $j^{th}$ position in $PC_i$ is obtained by $PC_i(j)$. 
The problem of determining the maximum memory requirement of all preemption chains with task $\tau_i$ as the base-task can then be formulated as an optimization problem.
The goal is to find the assignment of tasks $\in\Gamma_{\pi}$ to $PC_i$ such that:\\
$\forall j \in [1, |PC_i|-1]$:
\begin{equation}
    \label{eq:constrMem}
    \theta_{PC_i(j)} < P_{PC_i(j+1)}
\end{equation}
With the objective to maximize the memory requirement of all tasks in $PC_i$:
\begin{equation}
    \operatorname{Maximize:} \sum_{\forall \tau_j \in PC_i} M_j
\end{equation}


The solution to the optimization problem results in the preemption chain $PC_i$ for task $\tau_i$ with the maximum memory requirement. 
This analysis is similar to the task stack address calculation algorithms proposed for stack sharing in RTOS \cite{altmeyer2018empress}.

A task set is memory-feasible if the maximum preemption chain of each task on each core is not larger than $S$:
\begin{equation}
    \forall \pi_k \in \Pi, \forall \tau_i \in \Gamma_{\pi_k}: \sum_{\forall \tau_j \in PC_i} M_j \le S
\end{equation}

\section{Selecting Preemption Thresholds}
\label{sec:selectingPreemptionThresholds}

Several related works propose using preemption thresholds for better stack utilization in stack sharing systems and better schedulability \cite{wang1999scheduling, wang2016minimizing, saksena2000scalable, ghattas2007preemption}. 
Algorithms have been proposed to assign the maximum possible preemption thresholds, aiming to improve stack utilization. 
Ghattas et al. \cite{ghattas2007preemption} propose the Maximal Preemption Threshold Assignment Algorithm (MPTAA), which is commonly used in literature \cite{dong2014minimizing, wang2016minimizing, wang2016global}. 
MPTAA is based on the maximum preemption threshold algorithm proposed by Wang and Saksena \cite{wang1999scheduling, saksena2000scalable} and has been proven to find the largest preemption threshold assignment than any other feasible preemption threshold assignment \cite{chen2005solution}.
MPTAA starts from a schedulable initial preemption threshold configuration and increases the preemption thresholds to find a preemption threshold configuration with improved memory usage while maintaining the initial schedulability.
In addition, Wang and Saksena proposed a second algorithm for preemption threshold assignment, focusing on the schedulability aspect of the task set (see Figure 2 in \cite{wang1999scheduling}). We call this algorithm the preemption threshold assignment algorithm for schedulability (PTAS). While PTAS aims to improve schedulability using PTs, it does not aim to find a better solution from the memory perspective.
On the other hand, MPTAA aims to improve memory utilization using PTs; it does not improve schedulability. 
However, for a task set to be feasible, it must satisfy both timing and memory constraints, especially for 3-phase tasks. 
In this paper, we use the non-preemptive first preemption threshold assignment algorithm (NPFTA) proposed by Thilakasiri et al. \cite{thilakasiri_DATE}.
that aims to improve both schedulability and the memory aspect of task sets. NPFTA also reduces the time to find a solution and the explored state space.
All of the above algorithms, including NPFTA, are originally proposed for conventional tasks on single-core systems. 
Section \ref{Sec:NPFTA} provides background information on the NPFTA algorithm, and Section \ref{sec:adaptNPFTA} presents the adaptation of NPFTA to 3-phase tasks in a multi-core system.

\subsection{Background on NPFTA}
\label{Sec:NPFTA}

The non-preemptive first preemption threshold assignment algorithm, as the name suggests, starts with non-preemptive scheduling as its initial configuration and reduces the preemption thresholds of suitable tasks to achieve schedulability in necessary cases. Initially, the preemption thresholds of all tasks are configured to the maximum priority in the task set, i.e., non-preemptive. 
The primary reason for this is that non-preemptive scheduling gives the optimal memory utilization out of all non-preemptive, fully-preemptive, and limited preemptive configurations \cite{thilakasiri_DATE}.
The task sets schedulable under non-preemptive scheduling do not require any further changes to their preemption thresholds to achieve schedulability or improved memory usage.
Otherwise, the preemption thresholds of suitable tasks can be lowered, aiming for schedulability, unlike some of the existing algorithms, e.g., MPTAA.



\begin{algorithm}[h]
\caption{Non-Preemptive First Threshold Assignment Algorithm (NPFTAA)} 
\label{alg:overview}
\begin{algorithmic}[1]

\State \textbf{Input: } $\mathit{TaskSet}, \mathit{Cores}$ 

\State $\mathit{TaskSet} \gets  \mathit{assignNPPT()}$       \Comment{initial non-preemptive PTs}\label{line:NPPT}
\State $\mathit{memF \gets memReqAnalysis(TaskSet, Cores)}$ \label{Line:memAnalysis1}
\If{$\mathit{memF == True}$}                 \label{line:schedulabilityCheck_start}
    \State $\mathit{sched \gets WCRT\_Analysis(TaskSet)}$     \label{Line:WCRTAnalysis1}
    \If{$\mathit{sched == True}$}
        return $\mathit{Success}$
    \Else
                \State $\mathit{sched \gets Algorithm2(TaskSet, Cores)}$\label{line:PTA}
                \State $\mathit{memF \gets memReqAnalysis(TaskSet, Cores)}$ \label{Line:memAnalysis2}
                \If{$\mathit{sched==False || memF == False}$}
                    \State return $\mathit{Fail}$
                \Else 
                    \State return $\mathit{Success}$   
                \EndIf     
    \EndIf
\Else                   \label{line:schedulabilityCheck_end}
\State $\mathit{memF \gets False}$     \label{line:Np_noMemFitting}
\State return $\mathit{Fail}$

\EndIf
\end{algorithmic}
\end{algorithm}

\begin{algorithm}[h]
\caption{Recursively changing PTs aiming for schedulability } 
\label{alg:PTA}
\begin{algorithmic}[1]

\State \textbf{Input: } $\mathit{TaskSet}, \mathit{Cores}$ 

\For{$\mathit{core}$ in $\mathit{Cores}$}       \label{line:core_levelStart}
    \State $\mathit{missedTasks, otherTasks \gets sortTasks()}$
    \If{$\mathit{len(otherTasks) == 0}$}
        \State $\mathit{sched \gets False}$
        \State return $\mathit{Fail}$
    \ElsIf{$\mathit{len(missedTasks)> 0}$}               
        \State $\mathit{otherTasks.sort()}$       \Comment{sort according to the WCET}
        \State $\mathit{LPrio \gets getLowestPriority(missedTasks)}$   \Comment{find the lowest priority among the missed tasks}
        \State $\mathit{nextTask \gets getNextTask(otherTasks)}$
        \If{$\mathit{nextTask}$ not found}
            \State $\mathit{sched \gets False}$ 
            \State return $\mathit{sched}$
        \Else
            \State $\mathit{nextTask.setPT(LPrio - 1)}$ \label{line:lprio-1}
        \EndIf
    \EndIf     
\EndFor                                  \label{line:core_levelEnd} 
    
\State $\mathit{sched \gets WCRT\_Analysis(TaskSet)}$   \Comment{check schedulability after the changes}\label{Line:WCRTAnalysis2}
    \If{$\mathit{sched == True}$}
        \State return $\mathit{sched}$ 
    \Else \State $\mathit{Algorithm2(TaskSet, Cores)}$  \Comment{recurssive}\label{Line:recursive}
    \EndIf
    
\end{algorithmic}
\end{algorithm}

\subsection{Adaptation to 3-Phase tasks in a Multi-Core System}
\label{sec:adaptNPFTA}

In partitioned multi-core scheduling, tasks are statically assigned to cores, and preemption threshold assignment only affects tasks assigned to the same core. Thus, partitioned multi-core systems under preemption thresholds can be seen as equivalent to multiple single-core systems when it comes to preemption threshold assignment, i.e., the WCRTs of tasks assigned to remote cores are not affected by preemption threshold assignment of tasks on the local core.
However, this is only true if Equations \ref{EQ:caused_blocking}, \ref{Eq:B_mem}, and \ref{Eq:InterI_core} do not use the worst-case response times of remote tasks to bound the maximum number of their jobs that can cause interference or blocking on the local core. This is because while the use of the WCRTs to model delayed executions as release jitter helps to reduce pessimism compared to greedily assuming that there is always one carry-in job by each remote task (see Equation \ref{Eq:InterI}), it causes the WCRT of remote tasks to affect the WCRTs of the local core. 
Thus, partitioned cores in a multi-core system can't be considered equivalent to multiple single cores when assigning preemption thresholds when using WCRTs to model delayed executions of remote jobs. 
In this section, we adapt the NPFTA so that it can be safely used with WCRT analyses that are both independent and dependent on WCRTs of remote tasks.

The memory limitation is an even greater problem for systems with 3-phase tasks compared to conventional tasks due to the possibility of maximum memory footprints of multiple tasks needing to be accommodated in the limited local memory at a time. There exists a clear local memory constraint that needs to be met in order for the task set to be feasible on a given platform. 
This is accounted for in our adaptation. Any task set that is not memory-feasible under non-preemptive scheduling is concluded as infeasible and not evaluated further. This helps to eliminate infeasible task sets at an early stage.
For completeness, we present all the steps of the adapted NPFTA algorithm (see Algorithms \ref{alg:overview} and \ref{alg:PTA}).

Algorithm \ref{alg:overview} provides an overview of the adapted NPFTA algorithm to assign PTs and evaluate 3-phase tasks under preemption threshold scheduling.
The cores available on the platform and the task set, where each task is statically assigned to one of the cores, are inputs to the algorithm. It returns whether the task set can be successfully employed on the given platform, i.e., \emph{Success} if both schedulable and memory-feasible, \emph{Fail} if at least one of them is \emph{False}.
After the initial preemption thresholds are set to a non-preemptive configuration, i.e., set to the highest possible value (Line \ref{line:NPPT} in Algorithm \ref{alg:overview}), the task set is first evaluated from the memory perspective. The task-level memory requirement analysis proposed in Section \ref{sec:mem_requirement} is used for this purpose. 
If the task set does not meet the local memory constraints at the non-preemptive configuration, the task set is considered to be memory-infeasible (Line \ref{line:Np_noMemFitting} in Algorithm \ref{alg:overview}).
Otherwise, we continue to evaluate from the timing perspective (Lines \ref{line:schedulabilityCheck_start}-\ref{line:schedulabilityCheck_end} in Algorithm \ref{alg:overview}).
The WCRT analysis proposed in Sections \ref{sec:WCRT_taskPrio} and \ref{sec:WCRT_corePrio} for 3-phase tasks under preemption threshold scheduling is used to evaluate the taskset's schedulability (Line \ref{Line:WCRTAnalysis1} in Algorithm \ref{alg:overview}). 
In the case where the task set meets all deadlines as well under the non-preemptive configuration, the algorithm looks no further and returns $Success$.
In the case where the task set is not schedulable under non-preemptive scheduling, the algorithm looks further into configuring preemption thresholds to make the task set schedulable through a recursive process (see Algorithm \ref{alg:PTA}).
While this recursive process can be done at the core level for single-core systems and single-core equivalent partitioned multi-core systems, where the WCRTs are not dependent on WCRTs of remote tasks, to also support the proposed WCRT analysis, where the WCRTs on the local core are affected by WCRTs of tasks on remote cores, we do the recursive process at the task set level. Thereby, all the tasks that could possibly miss a deadline due to a preemption threshold change are identified.

For each core, a list of tasks that missed their deadline ($\mathit{missedTasks}$) and the rest of the tasks that met their deadlines ($\mathit{otherTasks}$) are identified.
If there are tasks that missed their deadline and there are tasks that met their deadlines (cases where not all tasks missed or met their deadlines), we reduce the preemption thresholds of suitable tasks that met the deadline to reduce blocking time on tasks that missed the deadlines, aiming to improve schedulability. 
At the core level, the algorithm first sorts the tasks that met their deadlines in descending order of their worst-case execution times (WCET).
Often, it is higher-priority tasks that miss the deadline in non-preemptive scheduling due to blocking from longer, lower-priority tasks. Thus, it is effective to reduce the PTs of lower-priority tasks that have longer execution times compared to the missed tasks, aiming to reduce the intra-core blocking.  
Another reason for this is that tasks with longer execution times often have longer periods and can often afford to be preempted without affecting their own schedulability.
In order to make the task set schedulable, even the task that has the lowest nominal priority ($\mathit{LPrio}$) among the tasks that missed the deadlines should meet its deadline. Thus, the PT of a task that has a lower nominal priority must be decremented below $\mathit{LPrio}$.
$\mathit{getNextTask()}$ finds the next task to reduce the PT of ($\mathit{nextTask}$). Which is the task with the longest WCET that has a higher or equal PT than $\mathit{LPrio}$ and a nominal priority lower than $\mathit{LPrio}$, i.e., $\mathit{\theta_{nextTask} \geq LPrio > P_{nextTask}}$, as the $\mathit{nextTask}$ to adjust the PT of. 
In other words, $\mathit{getNextTask()}$ finds the lower-priority task that the missed tasks couldn't preempt, which is potentially causing the most intra-core blocking.
In this process of finding the $\mathit{nextTask}$, we first look among the tasks that have a non-preemptive PT configuration (PT at the highest priority level). If no task among them satisfies the conditions, then we look at the tasks whose PTs have already been adjusted.
Once the $\mathit{nextTask}$ is found, we reduce its PT one level below the priority of the lowest priority task that missed the deadline (see Line \ref{line:lprio-1} in Algorithm \ref{alg:PTA}) so that preemption could help all missed tasks to meet their deadlines.
After each PT change, we evaluate the schedulability of the whole task set. If there are tasks that miss the deadline still, we repeat the process in Algorithm \ref{alg:PTA} recursively until the tasks are schedulable or until we identify that the task set is not schedulable, i.e., no more tasks found to adjust the PT of (Line \ref{Line:recursive} in Algorithm \ref{alg:PTA}).

\section{Evaluation}
\label{sec:Evaluation}

We evaluate the schedulability and memory-feasibility of 3-phase tasks under non-preemptive (NP), fully-preemptive execution phases (FP), and limited preemption (LP) using preemption threshold (PT) scheduling under both task-priority based memory-centric scheduler and a core-priority based memory-centric scheduler. 
In our plots, \emph{schedulable} refers to the percentage of task sets that meet all their deadlines, while \emph{sched\emph{+}mem} refers to the percentage of task sets that meet both deadlines and local memory constraints. 
Preemption threshold assignment can be coupled with priority assignment and mapping algorithms \cite{wang2016minimizing, wang2016global, Hatvani_Bril_PT}. 
However, in this paper, our focus is on evaluating how preemption thresholds improve the memory feasibility and schedulability of 3-phase tasks. 
Therefore, for all experiments, rate-monotonic priorities are used. Worst-fit mapping based on the memory footprint of tasks (referred to as WFM mapping from now on) is used for the task-priority based memory scheduler, and the heuristic mapping algorithm proposed by  Schwäricke et al. in \cite{schwaricke2020fixed}, which is a first-fit algorithm based on task periods and reduced processor capacity ($U/m$) is used for the core-priority based memory scheduler (referred to as FFP mapping from now on).
The task-priority-based memory-centric scheduler is referred to as TPMS, and the core-priority-based memory-centric scheduler is referred to as CPMS in the rest of the discussion and in the Figures \ref{fig:util_eval} and \ref{fig:SPM_eval}.
We assume non-negative integer timing parameters.

\subsection{Task Set Generation}
\label{Sec:taskSetGeneration}

The automotive period values reported by Kramer et al.~\cite{kramer2015real} $\{1, 2, 5, 10, 20, 50, 100, 200, \\ 1000\}$ms, with their associated probability percentages, $\{3, 2, 2, 25, 25, 3, 20, 1,  4\}$, respectively, are used to randomly generate minimum inter arrival times (MIAT). 
For a given total utilization $U$, the Dirichlet-Rescale (DRS) algorithm~\cite{DRS} is used to randomly generate $u_i$ utilization values for each task in the task set. 
Using the relationship, $C_i = u_i \cdot T_i$, the total execution time of the task is determined. 
The total data size of a task, $\mathit{Dm_i}$, is generated based on the distribution of the label sizes for an example engine control application provided by Kramer et al. in automotive applications~\cite{kramer2015real}, for a given number of labels ($l_i$) for the task. 
$l_i$ values for tasks are generated using a uniform distribution in the range [2, 100] labels per task.
Read data, $Rm_i$, and write data, $\mathit{Wm_i}$, sizes are calculated from $Dm_i$ using the percentages provided for read-only, write-only, and read-write label data in~\cite{kramer2015real}. 
Thus we consider, $Rm_i = 0.9\times Dm_i$ and $\mathit{Wm_i} = 0.6\times Dm_i$.
The code size of a task, $Im_i$, is generated using a uniform distribution in the range [2KB, 15KB]. 
The maximum stack usage of a task, $Ms_i$, is selected from a uniform distribution between [1KB, 4KB].

We consider that the read and write phase times are proportional to their corresponding data sizes. 
The ratio between the total time for memory phases and the execution phase, $\gamma$, is chosen randomly from the range [0.05, 0.15].
The total amount of data read during the read phase is considered as $Rm_i + Im_i$. The total amount of data written during the write phase is considered as $\mathit{Wm_i}$.
Therefore, the ratio between the read phase time, $C_i^r$, and the write phase time, $C_i^w$, can be expressed as $\alpha = (Rm_i + Im_i)/ \mathit{Wm_i}$. 
Also, we know the relationship $\gamma  \cdot C_i = C_i^r + C_i^w$. 
Using these two relationships, we get the write phase time, $C_i^w = C_i \cdot \gamma/(\alpha +1)$. Accordingly, the read phase time is calculated using $C_i^r = C_i \cdot \gamma - C_i^w$.
The time for the execution phase is calculated using the relationship, $C_i^e = (1- \gamma)\cdot C_i$.
We avoid the generation of inevitably unschedulable task sets, i.e., a lower priority task having a memory phase that is longer than the minimum inter-arrival time of a higher priority task, which makes the higher priority task unschedulable by nature due to mandatory inter/ intra-core blocking. 
In our experiments, a default configuration of $m=4, S=32\mathit{KB}$, and $n=32$ is used. 
The default total utilization value is $U=1$ for task-priority based memory-centric scheduling and $u=2.5$ for core-priority based memory-centric scheduling.
When one of these parameters is varied in an experiment, the rest of the parameters are kept constant at the default value.
Unique nominal priorities are assigned to tasks for the evaluation (unless mentioned otherwise, e.g., Figure \ref{fig:spm2}), similar to \cite{schwaricke2020fixed, arora2022analyzing}. The assigned preemption thresholds could be equal.
We also evaluate using minimum inter-arrival times that follow a log-uniform distribution in the range [10,1000]ms.

\subsection{Schedulability and Memory Feasibility Evaluation}
\label{Sec:feasibilityEvaluation}

\subsubsection{Varying the Total Utilization}
\label{sec:varyingUtil}

The total utilization of the task set is varied, and the schedulability and memory-feasibility are evaluated for automotive and log-uniform MIATs, with a TPMS, and a CPMS (Figure \ref{fig:util_eval}).
Although at different rates, schedulability decreases in all four experiments as utilization increases.
The common trend between all four experiments is that the NP scheduling provides the worst schedulability, and the FP scheduling provides the worst memory-feasibility. 
One of the key takeaways is that the LP version provides the best percentage of task sets that meet both timing and memory requirements (\emph{sched+mem}) for all four experiments.
The results showcase the drawbacks of the NP and FP ends, i.e., bad memory utilization in FP scheduling and schedulability limitations in NP scheduling, and how the limited preemption using PTs provides benefits of both ends, NP and FP.
\emph{sched+mem LP} provides upto 50\%, 80\%, 40\% and 25\% more schedulabile and memory-feasible task sets respectively in Figures \ref{fig:1} to \ref{fig:4} compared to \emph{sched+mem NP}. 
In addition to the memory-feasibility improvements resulted by the LP version, it is interesting to notice that the schedulability is also improved with the LP version over the FP scheduling with log-uniform MIATs (see Figures \ref{fig:2} and \ref{fig:4}). 
The schedulability for automotive task sets is maintained at the FP schedulability level. The proposed preemption threshold assignment algorithm finds PT configurations that makes the task set schedulable in the cases it isn't schedulable under NP scheduling. In this process NPFTA reduces the PT of certain tasks to reduce the blocking of tasks that misses the deadlines, i.e., tasks are removed from groups \circled{B} and \circled{F} when needed (see Section \ref{sec:selectingPreemptionThresholds}). 
The reason for why the schedulability of the LP version of automotive task sets is not improved over the FP schedulability is the existence of tasks with MIATs smaller than 10ms, i.e., 1ms, 2ms and 5ms.  These tasks are more sensitive to inter and intra-core blocking that can hinder their schedulability compared to the log-uniform tasks. 

 \begin{figure}

	\centering
 	\begin{subfigure}[b]{0.4\textwidth}
		\centering
		\includegraphics[width=\textwidth]{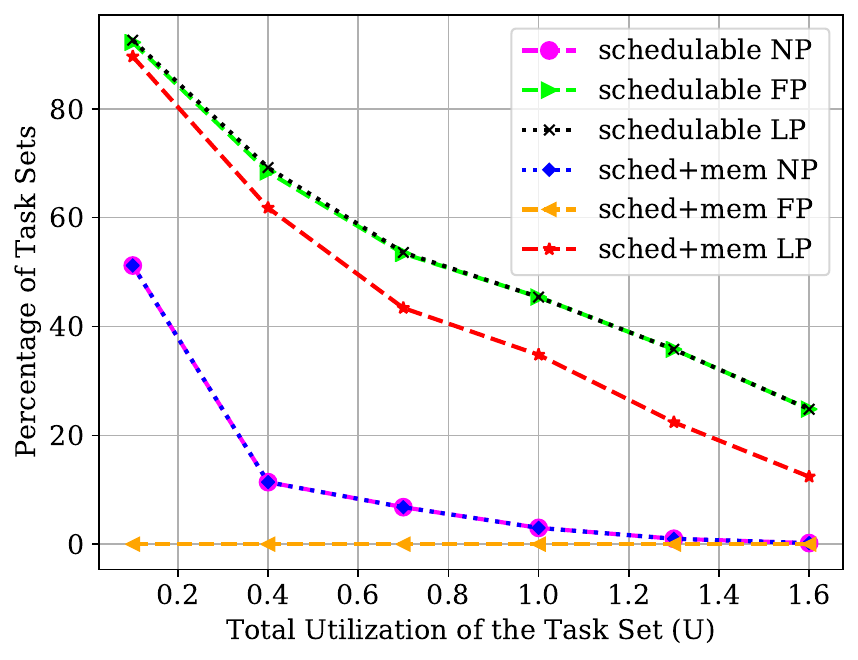}
		\caption{automotive periods with TPMS}
        \vspace{0.5em}
		\label{fig:1}
	\end{subfigure} 
    ~~~~
    \begin{subfigure}[b]{0.4\textwidth}
		\centering
		\includegraphics[width=\textwidth]{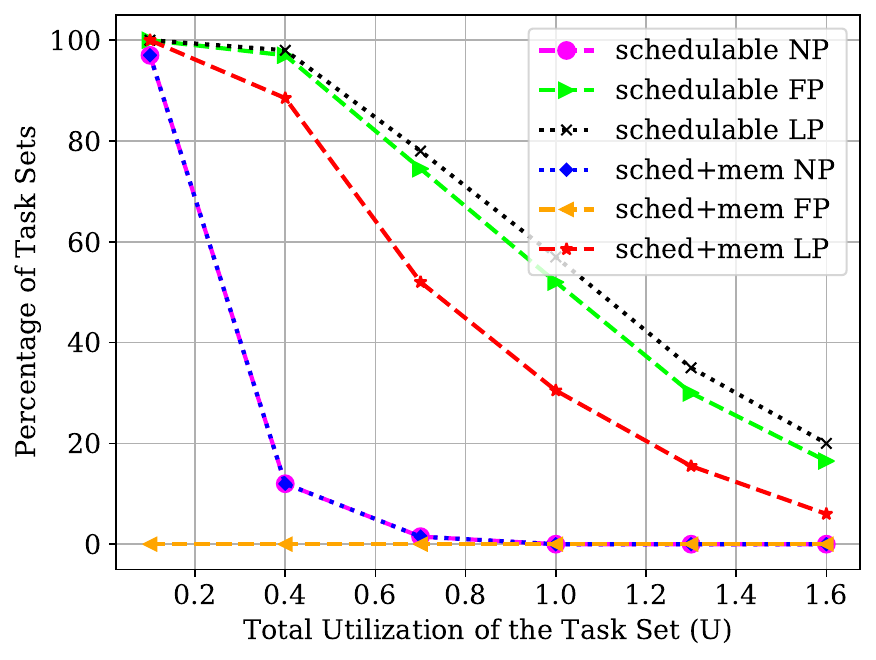}
		\caption{log-uniform periods with TPMS}
        \vspace{0.5em}
		\label{fig:2}
	\end{subfigure} 
    \begin{subfigure}[b]{0.4\textwidth}
		\centering
		\includegraphics[width=\textwidth]{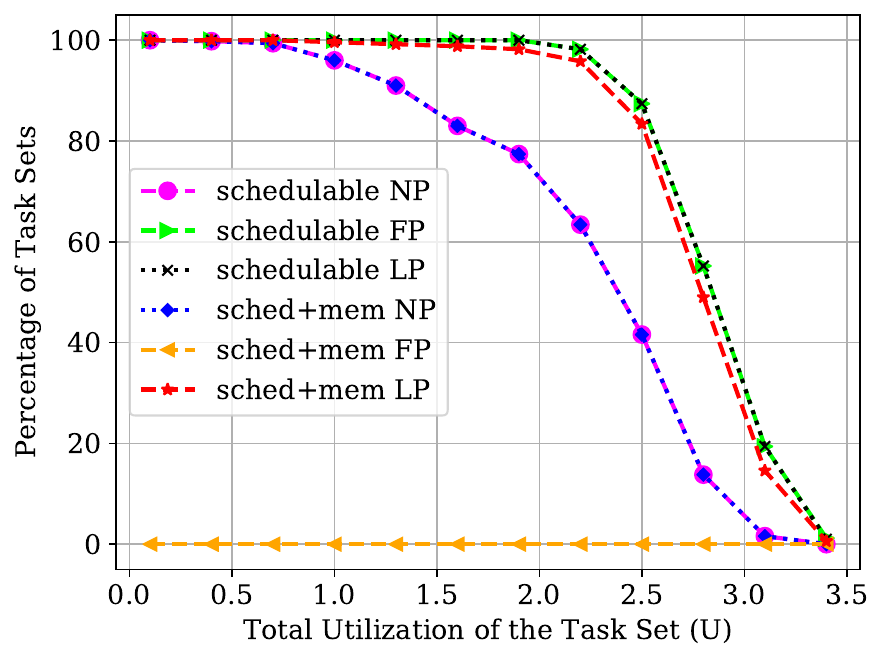}
		\caption{automotive periods with CPMS}
		\label{fig:3}
	\end{subfigure}
    ~~~~
    \begin{subfigure}{0.4\textwidth}
		\centering
		\includegraphics[width=\textwidth]{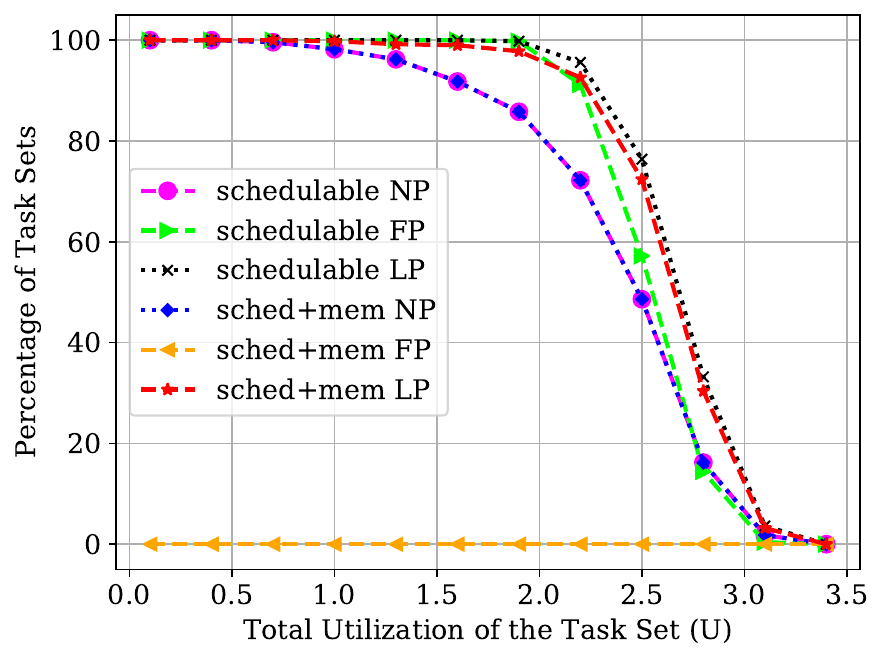}
            \caption{log-uniform periods with CPMS}
		\label{fig:4}
	\end{subfigure}
        \caption{Percentage of schedulable and memory-feasible task sets with a varying total utilization.}
        \label{fig:util_eval}
\end{figure}

The automotive tasks under CPMS perform the best with a varying total utilization, and the automotive tasks under TPMS perform the worst.
Another key insight is that both log-uniform and automotive tasks perform significantly better under a CPMS than under a TPMS.
This shows that for 3-phase tasks under preemption thresholds, a core-priority based memory-centric scheduler is more suitable than a task-priority based memory-centric scheduler.
This is an interesting outcome compared to the results presented by Arora et al. \cite{arora2022analyzing} for 3-phase tasks under limited preemptive scheduling using preemption points. 
In \cite{arora2022analyzing}, Arora et al. discussed different alternatives for memory-centric schedulers proposed in literature, i.e., offline \cite{pagetti2018automated}, TDMA-based memory-centric scheduler \cite{yao2012memory, tabish2019real}, and core-priority based memory-centric scheduler \cite{schwaricke2020fixed}. 
A task-priority based memory-centric scheduler is proposed in \cite{arora2022analyzing} for fully-preemptive E-phases and limited preemptive E-phases, where preemption is allowed at the boundary points of the E-phase.
Their results show that for fully-preemptive E-phases and limited preemptive E-phases with preemption points, a task-priority based memory-centric scheduler outperforms a core-priority based memory-centric scheduler in \cite{schwaricke2020fixed}. 
However, our results for limited preemptive scheduling using preemption thresholds show the contrary.
There are several reasons for this: The analysis of inter-core blocking becomes more complicated and possibly pessimistic for preemption thresholds due to the inter-core blocking that can be caused on a task in groups \circled{B} or \circled{F} (see Section \ref{sec:interB}) compared to that of in \cite{arora2022analyzing} for preemption points. Therefore, the analysis of inter-core memory blocking can bring down the schedulability of task sets.
Using the CPMS with memory phases that are preemptable by remote tasks entirely eliminates the inter-core blocking. A TPMS result from different task priorities being used to schedule the memory phases globally, while preemption thresholds are used to decide the preemptions locally. In CPMS, all the tasks assigned to a core have the same priority for scheduling memory phases, eliminating such complications; eliminating possible inter-core blockings on tasks in groups \circled{B} or \circled{F} (which is represented by the first term in Equation \ref{Eq:B_mem}). 
In addition, allowing memory phases to be preempted by higher-priority remote memory phases eliminates any other inter-core memory blocking caused by lower-priority remote memory phases.
Furthermore, using an appropriate mapping algorithm such as the FFP mapping proposed in \cite{schwaricke2020fixed} with a CPMS also reduces the inter-core interference on the highest-priority tasks, as they will be allocated to the core with the highest priority core that does not experience any inter-core interference.
Therefore, the highest-priority tasks assigned to the highest priority core experience zero inter-core interference and blocking, which boosts schedulability by a significant amount compared to the TPMS, especially due to the fact that the highest-priority tasks are the ones with the shortest inter-arrival times that are likely affected by the inter-core blocking with a TPMS.

\begin{figure}

	\centering
 	\begin{subfigure}[b]{0.32\textwidth}
		\centering
		\includegraphics[width=\textwidth]{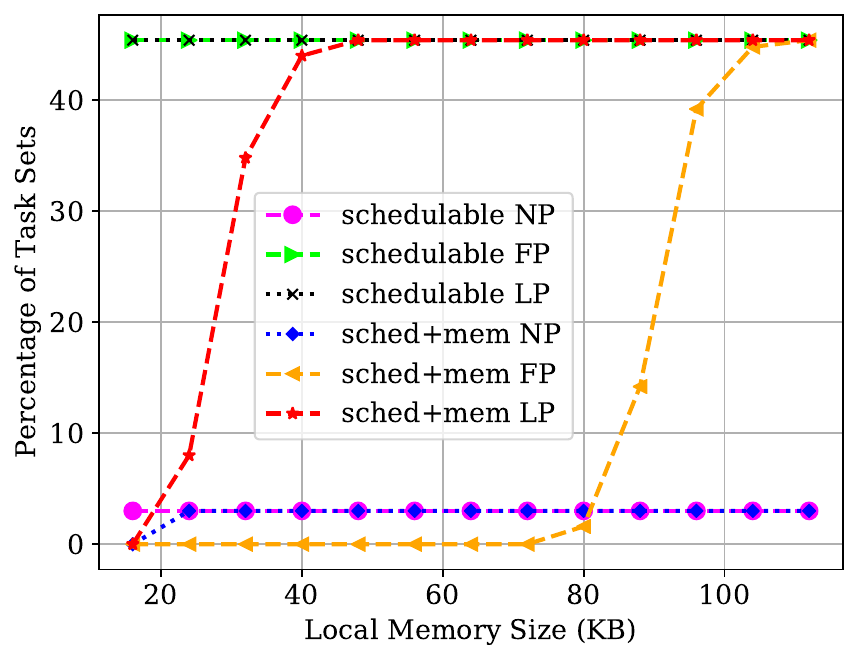}
		\caption{automotive periods with unique priorities and TPMS}
		\label{fig:spm1}
	\end{subfigure} 
    \begin{subfigure}[b]{0.32\textwidth}
		\centering
		\includegraphics[width=\textwidth]{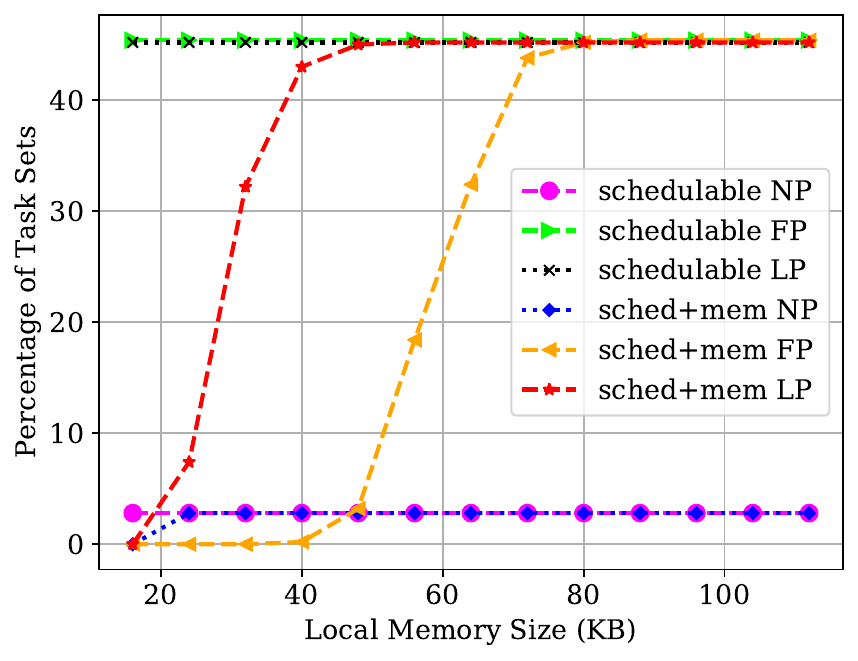}
		\caption{automotive periods allowing equal priorities with TPMS}
		\label{fig:spm2}
	\end{subfigure} 
    \begin{subfigure}[b]{0.32\textwidth}
		\centering
		\includegraphics[width=\textwidth]{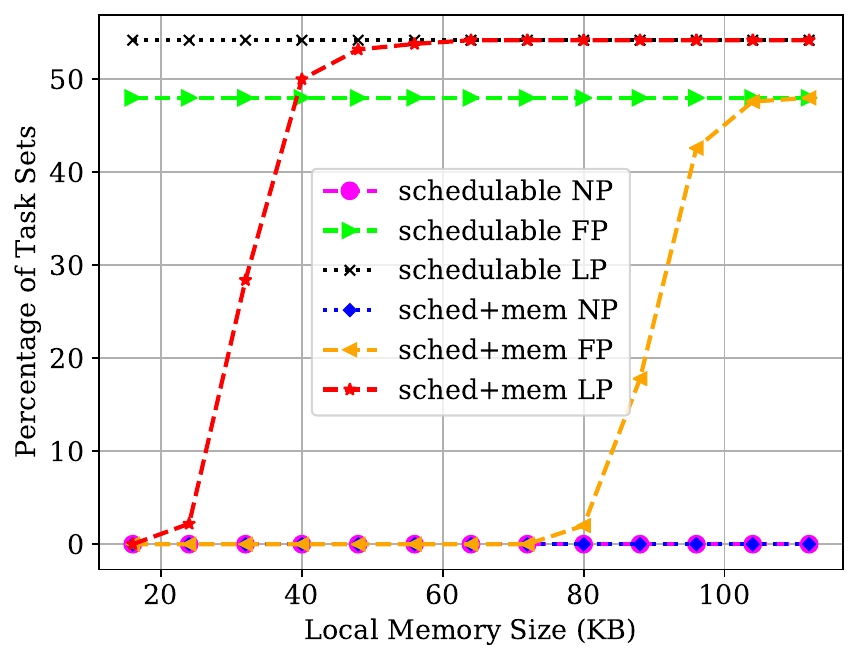}
		\caption{log-uniform periods with TPMS}
		\label{fig:spm3}
	\end{subfigure} 
    \begin{subfigure}[b]{0.32\textwidth}
		\centering
		\includegraphics[width=\textwidth]{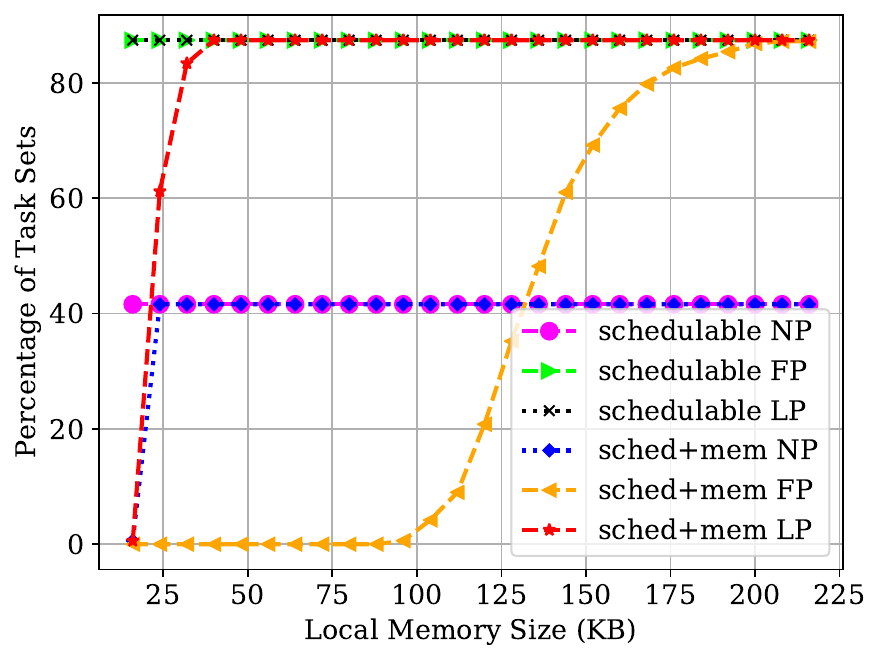}
		\caption{automotive periods with CPMS}
		\label{fig:spm4}
	\end{subfigure}
    \begin{subfigure}[b]{0.32\textwidth}
		\centering
		\includegraphics[width=\textwidth]{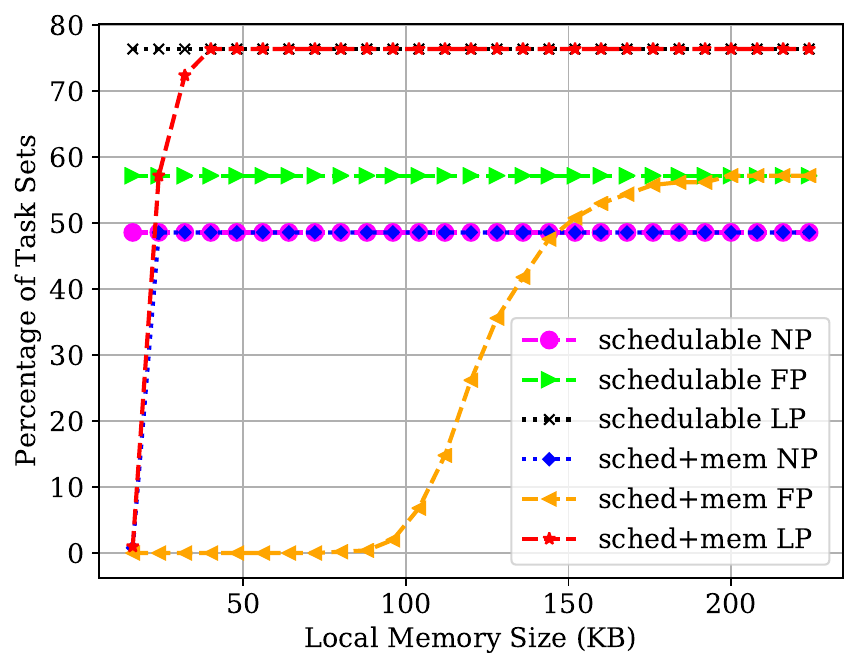}
		\caption{log-uniform periods with CPMS}
		\label{fig:spm5}
	\end{subfigure}
        \caption{Percentage of schedulable and memory-feasible task sets with a varying local memory size.} 
        \label{fig:SPM_eval}
\end{figure}

\subsubsection{Varying the SPM Size}
\label{Sec:varyingSPM}

In this section, we discuss the results from experiments that vary the local memory size (see Figure \ref{fig:SPM_eval}).
Among the five experiments, for the experiments using a TPMS, the total utilization is kept constant at $U=1.0$, and for a CPMS, it is kept at $U=2.5$. 
An interesting takeaway is that allowing equal priorities improves FP scheduling's memory utilization by $1.3\times$ compared to when unique priorities are used (see Figures \ref{fig:spm1} and \ref{fig:spm2}). This observation agrees with the concept of using non-preemptive groups to improve memory utilization, as tasks with the same priority are mutually non-preemptive \cite{altmeyer2018empress}. 

The schedulability is constant with the changing local memory size, as it is not affected by the SPM size.
The most interesting result from this experiment is the percentage of task sets that are both schedulable and memory-feasible. While all schedulable task sets with NP scheduling are memory-feasible after 24KB, the schedulability is low. 
In contrast, the schedulability results are good with FP scheduling. 
However, the memory feasibility is very low, i.e., zero until 72KB, 40KB, 72KB, 88KB, and 88KB, respectively, in the experiments presented in Figures \ref{fig:spm1} to \ref{fig:spm5}.
These results showcase the trade-off between FP and NP scheduling and the motivation for limited preemptive approaches.
In all five experiments, the limited preemptive scheduling using PTs provides the highest percentage of task sets that are schedulable and memory-feasible compared to both NP and FP scheduling.
For the schedulability improvement provided by the LP variant for log-uniform MIATs is clearly shown in Figures \ref{fig:spm3} and \ref{fig:spm5}.
The memory utilization of the task sets is improved with the PTs by $2.1\times$, $1.6\times$, $1.8\times$, $5\times$, and $5\times$, respectively, in the experiments presented in Figures \ref{fig:spm1} to \ref{fig:spm5} compared to the FP scheduling.
This shows the benefits of using PTs, especially when the local memory size is less than 104KB. 
Therefore, the proposed methods enable using a platform with a much smaller local memory for the same applications.




\section{Conclusions and Future Work}

This paper explores the use of preemption thresholds to limit the number of preemptions for 3-phase tasks under partitioned task-priority-based and core-priority-based memory-centric scheduling to increase memory feasibility while maintaining schedulability. 
Worst-case response time analyses to evaluate timing behavior and a worst-case memory requirement analysis to analyze memory feasibility are presented. 
A key takeaway from the results is that a core-priority-based memory-centric scheduler is more suitable and beneficial for the 3-phase task model under preemption thresholds.
Another key takeaway is that using preemption thresholds can significantly improve the percentage of task sets that meet both timing and memory constraints over both non-preemptive and fully-preemptive execution phases; significantly improving memory feasibility over fully preemptive scheduling while still maintaining or improving the schedulability compared to fully preemptive scheduling, resulting in task sets that are both memory feasible and schedulable up to 80\% more than non-preemptive scheduling.





\bibliography{ref}

\end{document}